\documentclass[UKenglish,10pt,conference]{IEEEtran}
\usepackage[utf8]{inputenc}
\usepackage[marginclue,footnote]{fixme}
\FXRegisterAuthor{jf}{ajf}{JF}
\FXRegisterAuthor{pw}{apw}{PW}
\FXRegisterAuthor{ls}{als}{LS}

\usepackage{amsmath}
\usepackage{amsfonts}
\usepackage{amssymb}
\usepackage{tikz-cd}
\usepackage{calrsfs}
\usepackage{hyperref}
\usepackage{bm}
\usepackage{makecell}
\usepackage{booktabs}
\usepackage{stmaryrd}
\usepackage{hyperref}

\usepackage{enumitem}

\usepackage{etoolbox}

\usepackage{tikz}
\usetikzlibrary{automata,positioning}

\apptocmd{\sloppy}{\hbadness 10000\relax}{}{}

\usepackage[appendix=append, bibliography=common]{apxproof}

\DeclareMathAlphabet{\pazocal}{OMS}{zplm}{m}{n}

\newcommand{\ptrace}{\mathsf{ptrace}}
\newcommand{\trinc}{\mathsf{trinc}}
\newcommand{\btop}{\mathsf{btop}}
\newcommand{\ndtop}{\mathsf{ndtop}}
\newcommand{\qrsim}{\mathsf{qrsim}}
\newcommand{\qsim}{\mathsf{qsim}}

\newcommand{\catC}{\mathcal{C}}
\newcommand{\catE}{\mathcal{E}}
\newcommand{\Set}{\mathbf{Set}}

\newcommand{\Pre}{\mathbf{Pre}}
\newcommand{\Eqv}{\mathbf{Eqv}}
\newcommand{\Top}{\mathbf{Top}}
\newcommand{\PMet}{\mathbf{PMet}}
\newcommand{\QMet}{\mathbf{HMet}}

\newcommand{\EM}{\mathit{EM}}
\newcommand{\Galg}[2]{\text{Alg}_#1(#2)}
\newcommand{\Str}{\textbf{Str}}
\newcommand{\str}{\mathbf{Str}}

\newcommand{\Pow}{\pazocal{P}}
\newcommand{\Prob}{\pazocal{D}}
\newcommand{\depth}{\mathsf{d}}

\newcommand{\Id}{\mathit{Id}}

\newcommand{\by}[1]{(\text{#1})}
\newcommand{\var}{V}
\newcommand{\act}{\mathsf{Act}}
\newcommand{\nat}{\mathbb{N}}
\newcommand{\id}{\mathit{id}}
\newcommand{\ar}{{\text{ar}}}
\newcommand{\sleq}{\sqsubseteq}
\newcommand{\mbar}{\overline{M}_1}

\newtheoremrep{thm}{Theorem}[section]
\newtheoremrep{prop}[thm]{Proposition}
\newtheoremrep{lem}[thm]{Lemma}
\newtheoremrep{cor}[thm]{Corollary}
\theoremstyle{definition}
\newtheoremrep{defn}[thm]{Definition}
\newtheoremrep{expl}[thm]{Example}
\newtheoremrep{rem}[thm]{Remark}

\tikzstyle{shiftarr}=[
    rounded corners,%
    to path={--([#1]\tikztostart.center)
                -- ([#1]\tikztotarget.center) \tikztonodes
                -- (\tikztotarget)},
]

\title{Conformance Games for Graded Semantics}

\author{\IEEEauthorblockN{Jonas Forster, Lutz Schr\"oder, Paul
  Wild}
  \IEEEauthorblockA{Friedrich-Alexander-Universität Erlangen-Nürnberg,
  Germany}
}

\begin{document}
\maketitle
\begin{abstract}
Game-theoretic characterizations of process equivalences
traditionally form a central topic in concurrency; for example, most
equivalences on the classical linear-time / branching-time spectrum
come with such characterizations. Recent work on so-called
\emph{graded semantics} has led to a generic behavioural equivalence
game that covers the mentioned games on the linear-time~/
branching-time spectrum and moreover applies in coalgebraic
generality, and thus instantiates also to equivalence games on
systems with non-relational branching type (probabilistic, weighted,
game-based etc.). In the present work, we generalize this approach
to cover other types of process comparison beyond equivalence, such
as behavioural preorders or pseudometrics. At the most general
level, we abstract such notions of \emph{behavioural conformance} in
terms of topological categories, and later specialize to
conformances presented as relational structures to obtain a concrete
syntax. We obtain a sound and complete generic game for behavioural
conformances in this sense. We present a number of instantiations,
obtaining game characterizations of, e.g., trace inclusion,
probabilistic trace distance, bisimulation topologies, and
simulation distances on metric labelled transition systems.
\end{abstract}

\section{Introduction}

\noindent Game-theoretic characterizations of equivalences have a firm place in
the study of concurrent systems. A well-known example is the classical
Spoiler-Duplicator game for
bisimilarity, which is played on pairs of states in labelled
transition systems (LTS), and in which Duplicator has a winning
strategy at a pair of states iff the two states are
bisimilar~\cite{Stirling99}. One benefit of such games is that they
provide witnesses for both equivalence and inequivalence, in the shape
of winning strategies for the respective player; from winning
strategies of Spoiler, one can in fact often extract distinguishing
formulae in suitable characteristic modal logics
(e.g.~\cite{KoenigEA20}). Besides the mentioned branching-time
bisimulation game, one has behavioural equivalence games for most of
the equivalences on the linear-time / branching-time spectrum of
process equivalences on LTS~\cite{Glabbeek01}. Similar spectra of
behavioural equivalences live over other system types beyond
relational transition systems, such as
probabilistic~\cite{JouSmolka90}, weighted, or neighbourhood-based
systems. We generally refer to equivalences on such spectra as the
\emph{semantics} of systems (bisimulation semantics, trace semantics
etc.).  Recently, a general treatment of behavioural equivalence games
has been given that works generically over both the system
type and the system semantics~\cite{DBLP:conf/lics/FordMSB022}. This
is achieved by on the one hand encapsulating the system type as a set
functor following the paradigm of universal coalgebra~\cite{Rutten00}
and on the other hand by abstracting the system semantics in the
framework of \emph{graded
  semantics}~\cite{DBLP:conf/calco/MiliusPS15}, which is based on
mapping the type functor into a graded monad~\cite{Smirnov08}.

Now the behaviour of concurrent systems can be compared in various
ways that go beyond equivalence. For instance, even classically,
(pre-)order-theoretic notions such as simulation or trace inclusion
have played a key role in system verification; and beyond two-valued
comparisons, there has been long-standing interest in behavioural
distances~\cite{GiacaloneEA90}. Following recent
usage~\cite{BeoharEA24}, we refer to such more general ways of
comparing systems as \emph{behavioural conformances}.  In a nutshell,
the contribution of the present work is to provide a generic
game-theoretic characterization of behavioural conformances,
parametrizing over the system type (given as a functor), the system
semantics (given as a graded semantics), and additionally the type of
behavioural conformance. Our technical assumption on behavioural
conformances is that they form a \emph{topological
  category}~\cite{AHS90} (essentially equivalently, a
$\textbf{CLat}_\sqcap$-fibration~\cite{DBLP:conf/lics/KomoridaKHKH19}). We
present two variants of the game, one where a finite number~$n$ is
determined beforehand and the game is then played for exactly~$n$
rounds, and one where the game is played for an infinite number of
rounds (or until a player gets stuck). Under the assumptions of the
framework, we show that the former characterizes the finite-depth
behavioural conformance induced by the given graded semantics, while
the latter characterizes an infinite-depth behavioural conformance
induced from the graded semantics under additional assumptions saying
essentially that no behaviour can be observed without taking at least
one evolution step in the system. Departing from the most general
setup, we subsequently refine the game under the assumption that the
topological category modelling the behavioural conformance is a
category of relational structures~\cite{DBLP:conf/calco/FordMS21},
which holds, for instance, for behavioural preorders and behavioural
(pseudo-)metrics. The coalgebraic codensity games pioneered by
Komorida et al.~\cite{DBLP:conf/lics/KomoridaKHKH19} are similarly
parametrized over the type of behavioural conformances via topological
categories / $\textbf{CLat}_\sqcap$-fibrations. One key difference
with our games is that codensity games so far apply only to the
branching-time case, while we are interested primarily in coarser
behavioural conformances on generalized linear-time / branching-time
spectra. That said, our generic games do apply also in the
branching-time case, and then instantiate to games that are markedly
different from codensity games. Indeed, our games are, roughly
speaking, dual to codensity games in that codensity games work with
modal \emph{observations} on states, while our games concern the way
behaviours of states are \emph{constructed} in an algebraic sense.

We apply this framework to a number case studies, obtaining
game-theoretic characterizations of classical trace inclusion of LTS;
bisimulation topologies; quantitative similarity of metric LTS; and
probabilistic trace semantics.

\paragraph*{Related Work} We have already discussed work on
coalgebraic codensity
games~\cite{DBLP:conf/lics/KomoridaKHKH19}. Similarly, Kupke and
Rot~\cite{KupkeRot21} give a coalgebraic treatment of
\emph{coinductive predicates} in terms of fibrations, generalizing in
particular behavioural distances but focusing on the branching-time
setting. By the fixpoint nature of behavioural equivalence, our
infinite game relates to some degree to general fixpoint
games~\cite{BaldanEA19,BaldanEA20}.\lsnote{@Jonas: Add details} Our
topological version of the powerset functor that appears in the case
study on bisimulation topologies owes ideas to the Vietoris topology
(via the use of hit sets), and more broadly to work on Vietoris
bisimulations~\cite{BezhanishviliEA10}. The treatment of spectra of
behavioural metrics via graded semantics goes back to work on
characteristic quantitative modal
logics~\cite{ForsterEA24,ForsterEA23}. This work joins a strand of
work on the coalgebraic treatment of behavioural distances
(e.g.~\cite{BreugelWorrell05,BaldanEA18,WildSchroder22}) and the
treatment of spectra of behavioural equivalences via graded
semantics~\cite{DBLP:conf/calco/MiliusPS15,DBLP:conf/concur/DorschMS19,DBLP:conf/lics/FordMSB022}. Graded
semantics essentially subsumes earlier coalgebraic treatments of
linear-time equivalences based on Kleisli~\cite{HasuoEA07} and
Eilenberg-Moore categories~\cite{JacobsEA15}, respectively. The
Kleisli and Eilenberg-Moore setups are alternatively subsumed by an
approach based on corecursive algebras~\cite{RotEA21}. Spectra of
behavioural metrics have been studied in a highly general approach
based on Galois connections, which, broadly speaking, subsumes a wide
range of examples but leaves more work to concrete
instances~\cite{BeoharEA23,BeoharEA24}. For the presentation of graded
monads on relational structures, we build on work on presenting
monads~\cite{DBLP:conf/lics/MardarePP16} and graded monads on metric
spaces~\cite{ForsterEA23} and
posets~\cite{DBLP:journals/mscs/AdamekFMS21,DBLP:conf/lics/FordMS21}
as well as generalizations to categories of relational
structures~\cite{DBLP:conf/calco/FordMS21,DBLP:phd/dnb/Ford23}.

\section{Preliminaries}
\noindent We recall background on graded
semantics~\cite{DBLP:conf/calco/MiliusPS15,DBLP:conf/concur/DorschMS19}
and topological categories~\cite{AHS90}.

\subsection{Topological Categories}

\noindent We perform most of our technical development on the level of
topological categories, covering examples like topological spaces,
measurable spaces, pseudometric spaces, preorders and equivalence
relations. The concept of topological category is based on the central
notion of initial lift, formally defined next. Generally, initial
lifts induce a structure on a set from a family of maps into sets
already carrying a corresponding structure. For instance, given a
set~$X$ and a family of maps $f_i\colon X\to Y_i$, $i\in I$, where
each of the sets~$Y_i$ is equipped with a pseudometric~$d_i$ (i.e.\ a
distance function that is like a metric except that distinct elements
may have distance~$0$, cf.\
Example~\ref{expl:topological-cats}.\ref{item:top-pmet}), we construct
the largest pseudometric~$d$ on~$X$ making all~$f_i$ non-expansive as
$d(x,y)=\bigwedge_{i\in I}d_i(f_i(x),f_i(y))$. The existence of~$d$
may be viewed as a property of the forgetful functor mapping a
pseudometric space $(X,d)$ to its underlying set~$X$. The general
setting is as follows.

\begin{defn}\label{def:top}
  Let $U\colon \catE\to \catC$ be a functor between categories
  $\catE$, $\catC$. A \emph{$U$-structured source} is a family
  $\mathcal S=(B\xrightarrow{f_i}UA_i)_{i\in I}$ of morphisms in
  $\catC$. Then, a \emph{$U$-initial lift} of~$\mathcal S$ consists of
  an $\catE$-object $A$ and a family of $\catE$-morphisms
  $(A\xrightarrow{\bar{f}_i}A_i)_{i\in I}$ such that
  $U\bar{f}_i = f_i$ for all $i\in I$, and for any source
  $(C\xrightarrow{g_i}A_i)_{i\in I}$ and map $h\colon UC\to UA$ such
  that $f_i\cdot h=Ug_i$ for all~$i$, there is an $\catE$-morphism
  $\bar h\colon C \to A$ such that $U\bar h=h$ and
  $\bar{f}_i \cdot \bar h = g_i$ for all~$i$. The functor~$U$ is
  \emph{topological}, and the pair $(\catE, U)$ (or, by abuse of
  terminology, just the category~$\catE$) is a \emph{topological
  category} if every $U$-structured source has a unique $U$-initial
  lift.
\end{defn}

\noindent
\begin{rem}
  If $(\catE,U)$ is topological, then~$U$ is faithful~\cite{AHS90}; in
  particular,~$h$ as in Definition~\ref{def:top} is necessarily unique
  (uniqueness of~$h$ is included in the definition of initiality
  in~\cite{AHS90} but not used in the proof of faithfulness).
  Uniqueness of initial lifts is not essential, and is ensured by the
  harmless condition that every identity-carried isomorphism is an
  identity.

  Every topological category is also \emph{co-topological}, i.e.\ has
  final lifts of costructured sinks (the dual notion of initial
  lifts). Topological categories are special kinds of fibrations, and
  indeed coincide (up to size issues not relevant here) with
  $\textbf{CLat}_\sqcap$-fibrations~\cite{DBLP:conf/lics/KomoridaKHKH19}.
  We generally use the terminology of topological categories, but
  borrow notation from the theory of fibrations when convenient.
\end{rem}

\begin{defn}
  Let $(\catE, U\colon \catE \to \catC)$ be a topological
  category. Every $\catC$-object $X$ induces a thin category $\catE_X$
  consisting of all $\catE$-objects $P$ such that $UP = X$ and all
  morphisms $h\colon P\to P'$ such that $Uh = \id_X$.  Every
  $\catC$-morphism $f\colon X \to Y$ then induces a \emph{reindexing functor}
  $f^\bullet\colon \catE_Y\to \catE_X$ that maps objects $P$ to the
  domain of the initial lift of the single-morphism $U$-structured
  source $f\colon X \to UP$.  We refer to the objects of $\catE_X$ as
  \emph{conformances} over~$X$. Since thin categories are the same
  as (possibly large)
  preorders, we employ order-theoretic notation for~$\catE_X$; in
  particular, we write $P\sqsubseteq P'$ if there is an
  $\catE_X$-morphism $P\to P'$, and we write~$\bigsqcap$
  and~$\bigsqcup$ for meets and joins, respectively, in~$\catE_X$
  (which exist because~$\catE$ is topological and co-topological).
\end{defn}

\noindent In particular, we can write the initial lift of a $U$-structured source
$(B\xrightarrow{f_i}UA_i)_{i\in I}$ as $\bigsqcap_{i\in I} f_i^\bullet
A_i$. We overload the reindexing notation to apply to morphisms
in $\catE$ as well, by defining of $f^\bullet := (Uf)^\bullet$ for any
$f$ in $\catE$. We will be working exclusively with categories that are
topological over the category $\Set$ of sets:

\begin{expl}\label{expl:topological-cats}
  \begin{enumerate}[wide]
  \item The category $\Eqv$ of equivalence relations and
    equivalence-preserving maps, that is, maps
    $f\colon (X,\sim_X)\to(Y,\sim_Y)$ such that $f(x) \sim_Y f(y)$
    whenever $x \sim_X y$.
    \item The category $\Pre$ of preorders (reflexive and transitive relations) and monotone maps.
    \item\label{item:top-pmet} The category $\PMet$ of (1-bounded)
      pseudometric spaces and nonexpansive maps, where
      $(X,d_X\colon X\times X\to[0,1])$ is a pseudometric space if
      $d_X(x,x) = 0$, $d_X(x,y)=d_X(y,x)$, and
      $d(x,z) \le d(x,y)+d(y,z)$ for all $x,y,z\in X$, and
      $f\colon(X,d_X)\to(Y,d_Y)$ is nonexpansive if
      $d_Y(f(x),f(y)) \le d_X(x,y)$ for all $x,y$.
    \item The category $\Top$ of topological spaces and continuous
      maps.
  \end{enumerate}
\end{expl}

\noindent The reindexing functor $f^\bullet\colon \catE_Y \to \catE_X$ corresponding to a map
$f\colon X\to Y$ has a left adjoint~$f_\bullet\colon \catE_X \to
\catE_Y$, called the
\emph{pushforward} or \emph{direct image} functor, which may be
computed as a final lift. We give explicit descriptions of
$f^\bullet$ and $f_\bullet$ for the topological categories $\Pre$
and $\PMet$, which we use in our running examples.

For ${\le_Y}\in\Pre_Y$ we have $x_1\, f^\bullet(\le_Y)\, x_2$ iff
$f(x_1)\, \le_Y\, f(x_2)$, which in particular means that
$f^\bullet(\le_Y)$ is the greatest (finest) preorder on $X$ that makes
$f$ monotonic.  The pushforward $f_\bullet({\le_X})$ of
${\le_X}\in\Pre_X$ is the least (coarsest) preorder on~$Y$ that makes
$f$ monotonic, constructed by putting
$f(x_1)\,f_\bullet({\le_X})\,f(x_2)$ for all $x_1,x_2$ such that
$x_1\,\le_X x_2$ and then closing under transitivity.

Similarly, for $d_X\in\PMet_X$ and $d_Y\in\PMet_Y$, the reindexing
$f^\bullet(d_Y)$ is the smallest distance on $X$ that makes $f$
nonexpansive, explicitly,
$f^\bullet(d_Y)(x_1,x_2) = d_Y(f(x_1),f(x_2))$, while the pushforward
$f_\bullet(d_X)$ is the largest distance on $Y$ that makes $f$
nonexpansive. (The roles of small and large are reversed here compared
to the previous example; intuitively speaking this is because in
$\PMet$, having smaller distance values amounts to having more
structure on the space).

\subsection{Universal Coalgebra}

\noindent We use \emph{universal coalgebra}~\cite{Rutten00} to achieve
genericity of our results with respect to the system type.

\begin{defn}
  Let $G\colon \catE \to \catE$ be a functor. A \emph{$G$-coalgebra} $(X,
  \gamma)$ consists of an $\catE$-object $X$ and a morphism $\gamma
  \colon X \to GX$. A \emph{homomorphism} between coalgebras $(X, \gamma)$
  and $(Y, \delta)$ is a $\catE$-morphism $h\colon X \to Y$ such that
  $Gh \cdot \gamma = \delta \cdot h$.
\end{defn}

\noindent In a $G$-coalgebra $(X,\gamma)$, we think of~$X$ as an
object of states, and of~$\gamma$ as a transition map that assigns to
each state a structured collection of successors, with~$G$ determining
the type of structured collections and hence the system type; these
ideas are illustrated concretely in Example~\ref{expl:coalgebras}. As
we are working with topological categories~$\catE$ over~$\Set$, we are
often interested in functors $G\colon\catE\to\catE$ that \emph{lift}
some functor $F\colon\Set\to\Set$, i.e.\ $UG = FU$.

On $\Set$ and similar categories, coalgebras admit a notion of
(branching-time) behavioural equivalence: States $x, y$ in a coalgebra
$(X,\gamma)$ are \emph{behaviourally equivalent} if there is a
homomorphism of coalgebras $h\colon X \to Y$ such that $h(x) = h(y)$
(behavioural equivalence between states in different coalgebras may be
captured via disjoint sums of coalgebras). We adapt this notion to the
setting of topological categories by defining the behavioural
conformance $P_\gamma^\infty$ on a coalgebra $(X,\gamma)$ as the
initial lift \(P_\gamma^\infty := \bigsqcap_h h^\bullet Y\) where
$h\colon X \to Y$ ranges over all coalgebra homomorphisms with domain
$(X, \gamma)$.\jfnote{Compare to definition via coinductive
  predicates. (Ours definition subsumes that one)}

Throughout the development, we use the trace inclusion preorder on
labelled transition systems and the probabilistic trace metric on
labelled Markov chains as running examples, illustrating concepts as
they are introduced.

\begin{expl}
  \label{expl:coalgebras}
  \begin{enumerate}[wide]
  \item \label{item:lts} We denote the finite powerset functor by
    $\Pow_\omega\colon \Set \to \Set$. Finitely branching labelled
    transition systems (LTS) over a set $\act$ of actions  are
    coalgebras for the functor $\Pow_\omega(\act \times {-})$.
    Behavioural equivalence then instantiates to the standard notion
    of bisimilarity.
  \item \label{item:pts} Markov chains with (weighted) edges labelled
    in a set $\act$ of actions are coalgebras for the functor
    $G = \Prob(\act \times {-})$, where the \emph{discrete
      distribution functor} $\Prob\colon \Set \to \Set$ maps a set~$X$
    to the set of finitely supported probability distributions over
    $X$, and acts on maps~$f$ via direct image:
    $\Prob f(\mu)(U) = \mu(f^{-1}[U])$. Behavioural equivalence in
    $G$-coalgebras coincides with probabilistic
    bisimulation~\cite{Klin09}.  The functor $\Prob$ can be lifted to
    $\overline \Prob \colon \PMet \to \PMet$ by equipping the set of
    distributions with the \emph{Wasserstein distance}: Given two
    distributions $\mu,\nu \in \Prob X$, define
    $\mathsf{Cpl}(\mu,\nu)$ to be the set of all \emph{couplings} of
    $\mu$ and $\nu$, that is, probability distributions
    $\rho\in\Prob(X\times X)$ such that $\Prob\pi_1(\rho) = \mu$ and
    $\Prob\pi_2(\rho) = \nu$.  Given a pseudometric $d_X$ on $X$, the
    Wasserstein distance is then given by
    \begin{equation*}
      d_{\overline \Prob(X,d_X)}(\mu,\nu) = \inf \{ \catE_\rho(d_X) \mid \rho\in\mathsf{Cpl}(\mu,\nu) \},
    \end{equation*}
    where
    $\catE_\rho(d_X) = \sum_{x_1,x_2\in X} \rho(x_1,x_2)\cdot
    d_X(x_1,x_2)$ denotes the expected distance of a pair of points
    under $\rho$.  The functor $G$ then lifts to
    $\overline G = \overline \Prob(\act \times {-})$, where~$\act$ is
    equipped with the discrete metric. The behavioural conformance
    $P^\omega_\gamma$ on a $\overline G$-coalgebra then coincides with
    bisimilarity distance~\cite{BreugelWorrell05,ForsterEA23b}.
  \end{enumerate}
\end{expl}

\subsection{Graded Monads and Graded Algebras}

\noindent We employ \emph{graded
  semantics}~\cite{DBLP:conf/concur/DorschMS19,DBLP:conf/calco/MiliusPS15}
to capture conformances that are coarser in nature than behavioural
equivalence, especially those that are \emph{linear-time} in the sense
of the linear-time/branching-time spectrum~\cite{Glabbeek01}.  This
framework is based on the central notion of \emph{graded
  monads}~\cite{Smirnov08}, which algebraically describe the structure
of observable behaviours at each finite depth, with \emph{depth}
understood as look-ahead, measured in terms of numbers of transition
steps. Formal definitions will be recalled presently; the guiding
example is the \emph{graded trace monad} on the category of sets,
which captures the trace semantics of $\act$-labelled transition
systems, modelled as $\Pow(\act\times(-))$-coalgebras. This graded
monad consists, inter alia, of set functors $M_n$ for $n<\omega$,
given by $M_nX=\Pow(\act^n\times X)$, with elements of $M_nX$
representing sets of traces annotated with a poststate. The
multiplication has type
$\Pow(\act^n\times\Pow(\act^k\times X))\to\Pow(\act^{n+k}\times X)$,
and distributes traces over sets in the expected way, thus embodying
the linear-time character of the graded monad. General definitions are
as follows.

\begin{defn}
  A \emph{graded monad} $\mathbb{M}$ on a category\;$\catC$ consists
  of a family of functors $M_n\colon \catC \to \catC$ for $n \in
  \nat$ and natural transformations $\eta\colon \Id \to M_0$ (the
  \emph{unit}) and $\mu^{n,k}\colon M_nM_k \to M_{n+k}$ for all $n,
  k \in \nat$ (the \emph{multiplications}), subject to essentially
  the same laws as ordinary monads up to the insertion of grades
  (specifically, one has unit laws $\mu^{0,n}\cdot \eta
  M_n=\id_{M_n}=\mu^{n,0}\cdot M_n\eta$ and an associative law
  $\mu^{n+k,m} \cdot \mu^{n,k}M_m=\mu^{n,k+m}\cdot M_n\mu^{k,m}$).
\end{defn}
\noindent In particular, $(M_0, \eta, \mu^{00})$ is an ordinary (non-graded)
monad.

\noindent Finally, the concept of monad algebra extends to the graded
setting.

\begin{defn}[Graded Algebra]
    Let $\mathbb{M}$ be a graded monad in $\catC$. A \emph{graded
    $M_n$-algebra} $((A_k)_{k \leq n}, (a^{mk})_{m+k \leq n})$ consists
    of a family of $\catC$-objects $A_i$ and morphisms $a^{mk}\colon
    M_mA_k \to A_{m+k}$ satisfying essentially the same laws as a
    monad algebra, up to insertion of the grades. Specifically, we
    have $a^{0m} \cdot \eta_{A_m} = \id_{A_m}$ for $m \leq n$, and
    whenever $m + r + k \leq n$, then $a^{m+r,k} \cdot
    \mu^{m,r}_{A_k}=a^{m,r+k} \cdot M_ma^{r,k}$. An
    \emph{$M_n$-homomorphism} of $M_n$-algebras $A$ and $B$ is a
    family of maps $(f_k\colon A_k \to B_k)_{k \leq n}$ such that
    whenever $m + k \leq n$, then $f_{m+k} \cdot a^{m,k}=b^{m,k} \cdot
    M_mf_k$. Graded $M_n$-algebras and their homomorphisms form a
    category $\Galg{n}{\mathbb{M}}$.
\end{defn}

\noindent (There is also a notion of graded $M_\omega$-algebras with
carriers $A_n$ for all~$n$~\cite{DBLP:conf/calco/MiliusPS15}, which
has formally analogous properties as the standard category of monad
algebras for a plain monad~\cite{FujiiEA16} but is not needed in the
current development.) It is immediate that
$((M_kX)_{k\leq n}, (\mu^{m,k})_{m+k \leq n})$ is an $M_n$-algebra for
every $\catC$-object~$X$. Again, $M_0$-algebras are just (non-graded)
algebras for the monad $(M_0, \eta, \mu^{00})$, and
$M_0$-homomorphisms are their morphisms.  Recall that together these
form a category $\EM(M_0)$, the \emph{Eilenberg-Moore category} of the
monad $M_0$.

Graded monads on $\Set$ can be presented via a graded variant of
equational theories~\cite{DBLP:conf/calco/MiliusPS15} (we will more
generally use graded relational-algebraic presentations in the present
work); most of the theory then needs to require that the operations
and equations concern only the next transition step, i.e.\ take place
in depth at most~$1$. This is abstractly captured by the following
notion~\cite{DBLP:conf/calco/MiliusPS15}:

\begin{defn}[Depth-1 Graded Monads]
  A graded monad is \emph{depth-$1$} if for all $n\in \mathbb{N}$, the
  morphism $\mu^{1,n}$ is a coequalizer of the following diagram:
\begin{equation}\label{eq:depth}
\begin{tikzcd}
  M_1M_0M_nX \arrow[r,swap, "\mu^{1,0}M_n", shift right] \arrow[r, "M_1\mu^{0,n}", shift left] & M_1M_nX \arrow[r, "\mu^{1,n}"] & M_{1+n}X.
\end{tikzcd}
\end{equation}
\end{defn}
\noindent (Mere commutation of the diagram is an instance of the laws
of graded monads.)
\begin{expl}\label{expl:graded-monads}
  We use the following graded monads on sets, which we will later lift to categories that are topological over $\Set$~(Example~\ref{expl:topological-cats}).
  \begin{enumerate}[wide]
  \item\label{item:gm-powers} Given a functor $F$, we have a depth-1 graded monad
    $M_n = F^n$, where $\eta$ and the $\mu^{nk}$ are identity.  This
    graded monad captures finite-depth behavioural
    equivalence~\cite{DBLP:conf/calco/MiliusPS15}.
  \item\label{item:gm-trace} As indicated above, the assignment
    $M_nX=\Pow_\omega(\act^n\times X)$ (for a set~$\act$ of actions)
    extends to a depth-1 graded monad, with multiplication given by
    distributing actions over powerset in the expected
    manner~\cite{DBLP:conf/calco/MiliusPS15}. In the paradigm of
    graded semantics as recalled next, it captures the trace semantics
    of labelled transition systems
    (Example~\ref{expl:coalgebras}.\ref{item:lts}).
  \item\label{item:gm-ptrace} We analogously have a depth-1 graded
    monad given by $M_nX=\Prob(\act^n\times X)$, where multiplication
    distributes actions over convex
    combinations~\cite{DBLP:conf/calco/MiliusPS15}. We will see that
    this graded monad captures the probabilistic trace semantics of
    labelled Markov chains
    (Example~\ref{expl:coalgebras}.\ref{item:pts}).
  \end{enumerate}
\end{expl}

\subsection{Graded Semantics}

\noindent Generally, a \emph{graded semantics}~\cite{DBLP:conf/calco/MiliusPS15}
of a functor $G\colon\catE \to \catE$ is given by a graded
monad~$\mathbb{M}$ and a natural transformation $\alpha\colon G \to
M_1$. Intuitively, $M_n1$ (where~$1$ is a terminal object of~$\catE$)
is a domain of observable behaviours after~$n$ steps,
with~$\alpha$ determining behaviours after one step. For a
$G$-coalgebra $(X, \gamma)$, we inductively define maps
$\gamma^{(n)}\colon X\to M_n1$ assigning to a state in~$X$ its
behaviour after $n$~steps: 

\begin{alignat*}{2}
    &\gamma^{(0)}&&\colon X \xrightarrow{M_0! \cdot \eta} M_01\\
    &\gamma^{(n+1)}&&\colon X \xrightarrow{\alpha \cdot \gamma} M_1X \xrightarrow{M_1\gamma^{(n)}}
    M_1M_n1 \xrightarrow{\mu^{1n}} M_{n+1}1
\end{alignat*}

\noindent In $\textbf{Set}$, we then obtain a notion of equivalence:
Two states $x_1,x_2 \in X$ in a coalgebra $\gamma\colon X \to GX$ are
behaviourally equivalent under the graded semantics
$(\alpha, \mathbb{M})$ if their behaviours agree at every depth,
i.e. if for all $k \in \nat$ we have
$\gamma^{(k)}(x_1) =\gamma^{(k)}(x_2)$. In many cases, one wants to equip
the space of behaviours with additional structure to cover, for
instance, notions of behavioural preorders or pseudometrics. We
capture such more general \emph{behavioural conformances} using
topological categories.

\begin{defn}
  Let $(\catE, U\colon \catE\to\catC)$ be a topological category. Let
  $\mathbb{M}$ be a graded monad on $\catE$ with functors
  $(M_n)_{n\in \nat}$.  The \emph{finite-depth behavioural
    conformance}~$P^\omega_\gamma$ of a $G$-coalgebra $(X,\gamma)$ with respect to
  a graded semantics $(\alpha, \mathbb{M})$ of~$G\colon \catE\to \catE$ is defined as (the
  object part of) the initial lift of the structured source
  \begin{equation*}
    (UX \xrightarrow{U\gamma^{(n)}} UM_n1)_{n\in \nat},
  \end{equation*}
  or, using reindexing notation, $P^\omega_\gamma = \bigsqcap_{n\in
  \nat} (U\gamma^{(n)})^\bullet M_n 1$.
\end{defn}

\noindent We continue to demonstrate these concepts by instantiating
them to our running examples:

\begin{expl}
  \label{expl:graded-semantics}
  \begin{enumerate}[wide]
  \item\label{item:trinc} Considering finitely branching LTS,
    i.e.~coalgebras for $G = \Pow(\act\times(-))$
    (Example~\ref{expl:coalgebras}.\ref{item:lts}), we characterize trace
    inclusion by lifting the graded monad
    from Example~\ref{expl:graded-monads}.\ref{item:gm-trace} to the
    graded monad $\mathbb{M}_\trinc$ on the topological category
    $\Pre$ of preorders and monotone maps: We put
    $M_n = \overline \Pow_\omega(\act^n \times (-))$, where $\act$ is
    considered a discrete preorder, and
    $\overline \Pow_\omega \colon \Pre \to \Pre$ is the lifting of
    $\Pow_\omega\colon \Set \to \Set$ that equips subsets of a
    preorder $(X,\leq)$ with the one-sided Egli-Milner ordering
      \begin{equation*}
        A \leq B \quad \iff \quad \forall a \in A.\; \exists b \in
        B.\; a \leq b.
      \end{equation*}

      \noindent In particular we have for $s,t\in M_n1$ that $s\leq t \iff s
      \subseteq t$.

      We then have a graded semantics $(\id, \mathbb{M}_\trinc)$ for
      the functor
      $\overline{G} = \overline \Pow_\omega(\act \times {-})$. For
      states $x_1,x_2$ in a $G$-coalgebra $(X,\gamma)$, we have
      $x_1\leq x_2$ in the behavioural conformance~$P_\gamma^\omega$
      iff for all $n\in \mathbb{N}$, the set of length-$n$ traces
      of~$x_1$ is a subset of all length-$n$ traces of~$x_2$,
      i.e. $\gamma^{(n)}(x_1) \leq \gamma^{(n)}(x_2)$.  This setup
      requires a $\overline{G}$-coalgebra; we convert an LTS, i.e.\ a
      $G$-coalgebra, into a $\overline{G}$-coalgebra by equipping its
      state set with the discrete (pre)order (which in particular
      ensures monotonicity of the transition map). In the sequel, we
      will generally convert set coalgebras into coalgebras in
      topological categories in this manner without further mention.
    \item \label{item:qtrace} For labelled Markov chains, i.e.\
      coalgebras for $\overline G = \overline \Prob(\act \times {-})$
      (Example~\ref{expl:coalgebras}.\ref{item:pts}), we similarly
      obtain a graded monad $\mathbb{M}_\ptrace$ on $\PMet$
      characterizing probabilistic trace distance by lifting the
      corresponding graded monad on set ($\Prob(\act^n\times(-))$,
      Example~\ref{expl:graded-monads}.\ref{item:gm-ptrace}) to
      $M_n = \overline \Prob(\act^{n} \times {-})$ where, again,
      $\act$ is equipped with the discrete
      pseudometric. %
      We then have a graded semantics
      $(\id, \mathbb{M}_\ptrace)$ for~$\overline G$. Given
      a state $x$ in a $\overline G$-coalgebra $\gamma$, the
      distribution $\gamma^{(n)}(x) \in \Prob(\act^{n})$ gives the probabilities of observing specific traces after $n$
      steps. To compare the observable behaviour of two states
      $x_1, x_2 \in X$, we can then measure how similar two
      distributions $\gamma^{(n)}(x_1), \gamma^{(n)}(x_2)$ are, by computing
      their total variation distance. The graded
      conformance $P^\omega_\gamma$ takes the supremum of these
      distances over all~$n$.
  \end{enumerate}
\end{expl}

\subsection{Pre-Determinization}\label{sec:pre-det}

\noindent From a depth-1 graded monad~$\mathbb{M}$ on a
category~$\catC$, we can canonically construct a functor~$\mbar$ on
the Eilenberg-Moore category of the
monad~$M_0$~\cite{DBLP:conf/lics/FordMSB022}. To begin, the definition
of graded algebras for~$\mathbb{M}$ implies that for $i\in \{0,1\}$,
we have functors
$({-})_i\colon \Galg{1}{\mathbb{M}} \to \Galg{0}{\mathbb{M}}$ taking
the $M_1$-algebra $(A_0, A_1, a^{0,0}, a^{0,1}, a^{1,0})$ to the
$M_0$-algebra $(A_i, a^{0,i})$. Now it can be shown that under
suitable conditions on~$\catC$, in particular for $\catC=\Set$, every
$M_0$-algebra~$A$ freely generates an $M_1$-algebra whose $0$-part
is~$A$. We obtain a functor
$E\colon \Galg{0}{\mathbb{M}}\to \Galg{1}{\mathbb{M}}$ that takes an
$M_0$-algebra~$A$ to the free $M_1$-algebra over~$A$. We then define
the functor
$\mbar \colon \Galg{0}{\mathbb{M}} \to \Galg{0}{\mathbb{M}}$ as
$\mbar := (E{-})_1$. The construction of~$EA$ involves a coequalizer
that in case of $M_0$-algebras of the form $M_nX$ instantiates
to~\eqref{eq:depth}, so that we have
$\mbar(M_nX, \mu^{0,n}) = (M_{n+1}X, \mu^{0,n+1})$.

Now let $L \dashv R$ be the free-forgetful adjunction of the category
$\Galg{0}{\mathbb{M}} \cong \EM(M_0)$.
Given a $G$-coalgebra $\gamma$, and a graded
semantics $(\alpha, \mathbb{M})$, we have
$X \xrightarrow{\alpha \cdot \gamma} M_1X = R\mbar LX$, which by
adjoint transposition yields an $\mbar$-coalgebra
$\gamma^\#\colon LX \to \mbar LX$ called the
\emph{pre-determinization} of $\gamma$.

\begin{expl}\label{expl:pre-det}
  In the case of traces in LTS
  (Example~\ref{expl:graded-monads}.\ref{item:gm-trace}), $M_0$ is simply the finite
  powerset functor, while $M_1 = \Pow_\omega(\act\times(-))$ is isomorphic to
  $(\Pow_\omega{-)}^\act$ when $\act$ is finite.
  The underlying map $R\gamma^\#\colon M_0X\to M_1X$ of the pre-determinization
  is therefore exactly equal to the classical \emph{powerset construction} which
  turns an LTS $X\to (\Pow_\omega X)^\act$ into a \emph{deterministic LTS}
  $\Pow_\omega X \to (\Pow_\omega X)^\act$.
\end{expl}

\noindent
Under the additional condition
that $M_01 \cong 1$, the coalgebra $\gamma^\#$ is an actual determinization in the
sense that on set-based coalgebras, (finite-depth) behavioural
equivalence of states $\eta(x_1)$ and $\eta(x_2)$ in $\gamma^\#$ coincides
with graded equivalence of $x_1$ and $x_2$ in $\gamma$. We take this as
motivation for the following definition:

\begin{defn}
Let $G$ be an endofunctor on a topological category. The
\emph{infinite-depth graded behavioural conformance} $P_\alpha^\infty$
on the $G$-coalgebra $(X, \gamma)$ under $(\alpha, \mathbb{M})$ is
defined as $P_\alpha^\infty = \eta^\bullet P^\infty_{\gamma^\#}$.
\end{defn}
\noindent That is, we take the behavioural conformance
$P^\infty_{\gamma^\#}$ on the pre-determinization $(LX,\gamma^\#)$ and
then reindex along the unit $\eta\colon X\to M_0X=RLX$.

\section{Topological Refinement}

\noindent We next establish some technical notions for algebras on topological
categories that will be relevant in the proofs of our main theorems.
For the remainder of this section, let $\catE$ be a topological
category over $\Set$ with forgetful functor $U\colon \catE \to \Set$,
and let $\overline T \colon \catE \to \catE$ be a lifting of a monad
$T \colon \Set \to \Set$.

\begin{defn}
A $\overline T$-algebra $p \colon \overline TP \to P$ is a
\emph{refinement} of a $\overline T$-algebra $(A, a)$ if
$P\in \catE_{UA}$ with $A \sleq P$ and the unique morphism
$\iota_{A, P}\colon A\to P$ with $U\iota_{A, P} = \id_{UA}$ is a
homomorphism.
\end{defn}
  \begin{lemrep}
    Let $a\colon \overline TA \to A$ be a $\overline T$-algebra, and let $P \in \catE$.
    Then there is at most one algebra structure $p \colon \overline TP \to P$
    such that $p$ is a refinement of $(A,a)$.
  \end{lemrep}
  \begin{proof}
    First note that $\overline T$ preserves epimorphisms, since $U$
    preserves and reflects, and any monad on $\Set$ preserves
    epimorphisms (e.g. \cite[Proposition 21.13]{DBLP:books/daglib/0023249}).
    Since $U\iota_{A,P} =
    \id_{UA}$ is an epimorphism, so is $\iota_{A,P}$, and by
    preservation through $\overline T$ also $\overline T \iota_{A,P}$.
    Let $p\colon \overline T P\to \overline P$
    and $p' \colon \overline TP \to P$ be refinements of $(A,a)$ with carrier $P$.
    Then since $\iota_{A,P}$ is a homomorphism, $p \cdot  \overline T\iota_{A, P}
    = \iota_{A,P} \cdot a = p' \cdot \overline T\iota_{A, P}$. Since
    $\overline T\iota_{A,P}$ is an epimorphism we have $p = p'$.
  \end{proof}

  \noindent We thus refer to elements of $\catE_{UA}$ as
  \emph{refinements} of $(A,a)$ if they carry such an algebra, and
  leave the algebra structure implicit.
  \noindent The refinement of $a$ \emph{generated} by
  $Q\in \catE_{UA}$ is the least refinement $C$ of $a$ such that
  $Q \sleq C$.

  \begin{lemrep}
    The refinement of an
    algebra $a\colon \overline T A\to A$ generated by $Q \in \catE_{UA}$ exists
    uniquely.
  \end{lemrep}
  \begin{proof}
    \label{lem:generated-cong}
    Let $S$ be the set of all refinement $P$ such that $Q \sleq P$,
    and let $C = \bigsqcap S$. Then $C$ is the limit of the diagram of
    all $\iota_{P, P'}$, for $P\sleq P'$ with $P,P' \in S$. Let $p$
    denote the algebra structure on the refinement $P$. The
    morphisms $p \cdot \overline T \iota_{ C, P}\colon \overline T C \to P$ form a cone of this
    diagram, which thus factors through~$C$ via a unique morphism
    $c \colon \overline T C \to C$. It is easy to see that $(C, c)$ is a
    refinement of $(A,a)$, and by construction $C$ is below any other
    refinement $P$ with $Q \sleq P$.
  \end{proof}

  \noindent We denote the refinement generated by a conformance $Q$ by
  $C(Q)$ (when the algebra is clear from context).

In the proof of our main results, we additionally need the following
claim, stating that reindexing under a homomorphism of algebras yields
a refinement.

  \begin{lemrep}
    \label{lem:hom-reindxing}
    Let $(A, a)$ and $(B,b)$ be $\overline T$-algebras, and let
    $h\colon A \to B$ be a homomorphism. Then $h^\bullet B$ is a
    refinement of $(A,a)$.
  \end{lemrep}
  \begin{proof}
    Consider the following diagrams:

    \begin{equation*}
    \begin{tikzcd}
      U\overline T(h^\bullet B) = U\overline TA  \arrow[r, "U\overline
      Th"] \arrow[d, "Ua"] & U\overline TB \arrow[d, "Ub"] \\
      Uh^\bullet B = UA \arrow[r, "Uh"] & UB
    \end{tikzcd}
    \qquad
    \begin{tikzcd}
      \overline T(h^\bullet B)  \arrow[r, "\overline Th"] \arrow[d,
      dashed, "a'"] & \overline TB \arrow[d, "b"] \\
      (h^\bullet B) \arrow[r, "h"] & B
    \end{tikzcd}
  \end{equation*}

  The left hand diagram commutes by hypothesis. Since $h$ is initial
  in the right hand diagram, we obtain the dashed morphism in the
  right hand diagram by initiality. It is obvious that $\iota$ is a
  homomorphism of the $\overline T$-algebras $(A, a)$ and $(h^\bullet B, a')$.
  \end{proof}

  \noindent In order for~$\mbar$ as in \autoref{sec:pre-det} to exist,
  we need the category $\EM(M_0)$ to have coequalizers, which we
  establish under the mild condition that~$M_0$ lifts some set
  functor. To give a description of the colimits, we use the
  following lemma:

  \begin{lemrep}
    The category $\EM(\overline T)$ is topological over $\EM(T)$.
  \end{lemrep}
  \begin{proof}
    Let $V\colon \EM(\overline T) \to \EM(T)$ denote the functor that
    forgets the conformance on an algebra, and let $W \colon \EM(T)
    \to \Set$ be the functor that forgets the algebra structure. Let
    $S$ be a structured source of algebra homomorphisms $(f_i\colon
    (A, a) \to V(B_i, b_i))_{i\in I}$ in $\EM(T)$. Let $\overline A$
    denote the object in $\catE$ that carries the initial lift of the
    structured source $(f_i \colon A \to UB_i)_{i\in I}$.

    We first show that $\overline A$ carries the algebra
    structure $a$, i.e. that $a \colon \overline T \overline A \to
    \overline A$ is a morphism in $\catE$. Since $b_i \cdot \overline
    T f_i$ are all morphisms in $\catE$, and $f_i \cdot a = Tf_i \cdot
    b_i$ in $\Set$, this follows by initiality of $\overline A$.

    It now remains to show that
    $(\overline A, a)$ is an initial lift of the source $S$: Let $(g_i
    \colon (C, c) \to (B_i, b_i)$ be a source of homomorphisms in
    $\EM(\overline T))_{i\in I}$, and let $h \colon V(C, c) \to (A, a)$ be a
    homomorphism of algebras such that $f_i \cdot h = g_i$. By
    initiality of $\overline A$ in $\catE$, we have that $h \colon
    \overline C \to \overline A$ is a morphism in $\catE$.
  \end{proof}

  \noindent The proof shows that initial lifts in $\EM(\overline T)$
  are constructed like in~$\catE$. We then have that colimits in
  $\EM(\overline T)$ can be constructed by taking the colimit in
  $\EM(T)$ and equipping it with the conformance of the final lift.

  \begin{correp}
    The category $\EM(\overline T)$ is cocomplete.
  \end{correp}
  \begin{proof}
    This follows from the facts that $\EM(T)$ is cocomplete
    \cite[20.34]{AHS90}, and that
    categories that are topological over a cocomplete category are
    cocomplete \cite[Theorem 21.16]{AHS90}.
  \end{proof}

\section{Graded Conformance Games}

\noindent We proceed to present the main object of interest,
\emph{graded conformance games}, which characterize a given graded
behavioural conformance.  For the remainder of the technical
development, fix a coalgebra $\gamma \colon X \to GX$ for a functor
$G\colon \catE \to \catE$ on a topological category $\catE$ over
$\Set$, equipped with a depth-1 graded semantics
$(\alpha, \mathbb{M})$ such that~$M_0$ is a lifting of some set
functor.

Roughly speaking, the game is played on the fibre of~$\catE$ over the
carrier of the pre-determinization of $(X,\gamma)$, i.e.\ on the
lattice $\catE_{UM_0X}$. We need the following lattice-theoretic
notion:

\begin{defn}
  A subset $B \subseteq L$ of a complete (join semi-) lattice~$L$ is
  \emph{generating} if for every $P \in L$, there is $A \subseteq B$
  such that $P = \bigvee A$.
\end{defn}

\noindent Now fix a generating set $B$ for the complete lattice
$\catE_{UM_0X}$, whose elements we shall call \emph{basic
  conformances}.

The game is played by two players, Spoiler and Duplicator, where
Duplicator aims to establish that a given basic conformance $P \in B$ is
under the graded behavioural conformance.  To this end, Duplicator
plays a set $Z\subseteq B$ of basic conformances on $UM_0X$, with
$Q:= \bigsqcup Z$, such that the inequality
$P \sleq (\mbar \iota \cdot \gamma^\#)^\bullet \mbar C(Q)$ holds,
where $\iota = \iota_{M_0X,Q}$ (we
say that $Z$ is \emph{admissible} in position $P$ if this condition is
satisfied). In their turn, Spoiler then chooses a conformance
$P' \in Z$ as the next position. If a player is unable to make a move,
they lose. After $n$ rounds, Duplicator wins if for the current position
$P$ it holds that $P \sleq (M_0!_X)^\bullet M_01$. We refer to this
last check as \emph{calling the bluff}.

\begin{table}[h!]
  \centering
  \setlength{\tabcolsep}{3pt}
  \renewcommand{\arraystretch}{1.3}
  \begin{tabular}{l | l | l}
    \thead{Player} & \thead{Position} & \thead{Move}\\
    \hline
    Duplicator & $P \in B$ & $Z \subseteq B$ s.t. $Z$ admissible\\
    Spoiler & $Z \subseteq B$ & $P \in Z$\\
  \end{tabular}
  \vspace{2mm}
  \caption{The graded conformance game}
\end{table}

\begin{rem}
  We understand the basic conformances as representing a complete set
  of properties~$P$ that the behavioural conformance~$P^\omega_\gamma$
  may or may not satisfy, in the sense that
  $P\sqsubseteq P^\omega_\gamma$. For instance, when conformances are
  equivalence relations, then we may take equivalence relations
  identifying precisely two elements as the basic conformances, and
  given such a basic conformance~$P$ identifying precisely~$x$
  and~$y$, the behavioural conformance is above~$P$ iff it
  identifies~$x$ and~$y$ (among others). In this view, a Duplicator move
  from~$P\in B$ to~$Z\subseteq B$ amounts to Duplicator claiming that
  the behavioural conformance satisfies all properties in~$Z$, and
  this move is admissible if the properties in~$Z$ imply~$P$ when one
  looks ahead precisely one transition step. Spoiler is then allowed
  to challenge any of the properties in~$Z$.  Komorida et.
  al~\cite{DBLP:conf/lics/KomoridaKHKH19} distinguish trimmed and
  untrimmed versions of their dual codensity games, where in the
  untrimmed variant, Duplicator may play any conformance on the state
  space, while the trimmed variant restricts them to playing only
  conformances from the generating set.
  One might indeed imagine an untrimmed variation of our game
  as well, with Duplicator first playing any conformance $P$ (not
  necessarily from the generating set) that satisfies the admissibility
  condition and Spoiler then playing any conformance $P' \sleq P$
  below it.  This is equivalent to setting
  $B = \catE_{M_0X}$. The involvement of Spoiler in this untrimmed game is
  completely unnecessary, since Spoiler could always challenge the
  entire conformance $P$ by playing $P' = P$. This is always a
  rational strategy for Spoiler, and thus the above game may be
  rephrased as a game played by a single player Duplicator, where in
  position $P$ they must play an admissible move $P'$, and then
  continues from the position $P'$.
\end{rem}

\begin{defn}
  The \emph{graded Kleisli star} for a morphism $f \colon X
  \to M_kY$ and $n \in \nat$ is given by
  \[f^*_n := (M_nX \xrightarrow{M_nf} M_nM_kY \xrightarrow{\mu^{n,k}_Y}
  M_{n+k}Y)\]
\end{defn}

\noindent Note that the underlying morphism of the pre-determinization $
\gamma^\#$ can equivalently be written as $(\alpha \cdot \gamma)^*_0$.
The proof of our first main theorem uses the following
lemma:
\begin{lemrep}
  \label{lem:mBarKleisli}
  Let $\mathbb{M}$ be a depth-1 graded monad on $\catC$. Then
  $\mbar(f^*_0) = f^*_1$ for every $\catC$-morphism
  $f\colon X \to M_kY$.
\end{lemrep}
\begin{proof}
  The fact that $f^*_n$ is a homomorphism of algebras follows from the
  graded monad axioms and naturality of the multiplications.

  By definition, $\mbar(f^*_0) = h$ where $h\colon M_1X\to M_{k+1}Y$
  is the unique morphism (obtained by canonicity of $(M_0X, M_1X)$)
  such that $(f^*_0, h)$ is a homomorphism of $M_1$-algebras. Thus it
  suffices to show that $(f^*_0, f^*_1)$ is a homomorphism of
  $M_1$-algebras. To this end, we show the remaining condition
  \begin{equation*}
    \mu^{1,k}_{M_{1+k}Y} \cdot M_1(f^*_0) = f^*_1 \cdot \mu^{1,0}_X
  \end{equation*}
  by the calculation

  \begin{align*}
    \mu^{1,k}_{M_{k}Y} \cdot M_1(f^*_0)
      &= \mu^{1,k}_{Y} \cdot M_1\mu^{0,k}_Y \cdot M_1M_0 f& \by{def. of $(-)^*_0$}\\
      &= \mu^{1,k}_{Y} \cdot \mu^{1,0}_{M_kY} \cdot M_1M_0f &
      \by{Graded monad axioms}\\
      &= \mu^{1,k}_{Y} \cdot M_1f \cdot  \mu^{1,0}_{X} &
      \by{Naturality of $\mu^{1,0}$}\\
      &= f^*_1 \cdot \mu^{1,0}_X & \by{def. of $(-)^*_1$}
  \end{align*}

\end{proof}

\begin{thm}
  \label{thm:main-fin} Let $(\alpha, \mathbb{M})$ be a depth-1 graded
  semantics for a functor $G\colon \catE \to \catE$ such that~$\mbar$
  preserves initial arrows whose codomain is of the form $M_i1$, and
  let $(X, \gamma)$ be a $G$-coalgebra. For $n<\omega$, Duplicator wins
  a position $P \in B$ in the $n$-round graded conformance
  game iff $\eta^\bullet P \sleq P^\omega_\gamma$.
\end{thm}
\begin{proof}
  Let $W_n\subseteq B$ be Duplicator's winning region in the $n$-round
  graded conformance game~$\mathcal{G}_n(\gamma)$. We show that
  $P \in W_n$ if and only if
  $P \sleq ((\gamma^{(n)})^*_0)^\bullet M_n1$, by induction on
  $n$. For $n=0$, we need to show that `calling the bluff' is sound and
  complete, that is, for a conformance~$P$ on $M_0X$, we have 
  $P \sleq (M_0!_X)^\bullet M_01$ iff
  $P \sleq ((\gamma^{(0)})^*_0)^\bullet M_01$. This follows by
  the following calculation:
  \begin{equation*}
      \begin{split}
      (\gamma^{(0)})^*_0&= \mu^{0,0} \cdot
      M_0\gamma^{(0)}\\ &=\mu^{0,0}\cdot M_0(M_0!_X\cdot
      \eta_X)\\ &=\mu^{0,0}\cdot M_0M_0!_X\cdot M_0\eta_X\\ &=M_0!_X
      \end{split}
  \end{equation*}
  where the last step follows from commutativity of the following
  diagram:
  \begin{equation*}
      \begin{tikzcd}
          M_0X \arrow[r, "M_0\eta_X"] \arrow[rd, "id_{M_0X}"'] & M_0M_0X
          \arrow[r, "M_0M_0!"] \arrow[d, "{\mu^{0,0}_X}"] & M_0M_01 \arrow[d,
          "{\mu^{0,0}_1}"] \\ & M_0X \arrow[r, "M_0!"]  & M_01
      \end{tikzcd}
  \end{equation*}
  \noindent For the inductive step, assume that $P \in W_k$ if and only
  if $P \sleq ((\gamma^{(k)})^*_0)^\bullet M_k1$.

  `$\Leftarrow$': Let~$P$ be a conformance on $M_0X$ such
  that $P \sleq ((\gamma^{(k+1)})^*_0)^\bullet M_{k+1}1$. We need
  to show that Duplicator has a winning strategy at $P$ in the
  $(k+1)$-round game. Since~$B$ is generating, we can pick
  $Z \subseteq B$ such that
  $\bigsqcup Z = ((\gamma^{(k)})^*_0)^\bullet M_k1$ as Duplicator's
  move. We write $Q := \bigsqcup Z$. Note that $Q$ is a refinement on
  $M_0X$ by Lemma~\ref{lem:hom-reindxing}. For every possible subsequent
  Spoiler move $P' \in Z$, we
  then have immediately that
  $P' \sleq ((\gamma^{(k)})^*_0)^\bullet M_k1$, so by induction
  hypothesis $P' \in W_k$. Therefore we just need to show that $Z$ is
  an admissible move for Duplicator, i.e. we want to show that
  $P \sleq (\mbar\iota_{Q} \cdot \gamma^\#)^\bullet \mbar
  C(Q)$.  Since by definition of $Q$, the arrow
  $(\gamma^{(k)})^*_0: Q \to M_k1$ is
  initial  (where $(\gamma^{(k)})^*_0: M_0X \to M_k1$ factors
  through $Q$ via $\iota_{Q}$ ), and $\mbar$ preserves that initial
  arrow,
  $\mbar(\gamma^{(k)})^*_0$ is also initial.  The top diagram
  \begin{equation*}
    \begin{tikzcd}[column sep=large, row sep=large]
      UM_0X  \arrow[r, "U\gamma^\#"] \arrow[rrd,
      "U(\gamma^{(k+1)})^*_0"'] & U\mbar M_0X \arrow[r,
      "U\mbar \iota_{Q}"] & U\mbar Q \arrow[d, "U\mbar(\gamma^{(k)})^*_0"] \\
      & &U \mbar M_k1
      &\\
      P  \arrow[rr, dashrightarrow, "\exists!"]
      \arrow[rrd, "(\gamma^{(k+1)})^*_0"'] & & \mbar Q \arrow[d, "\mbar(\gamma^{(k)})^*_0"] \\
      & &\mbar M_k1                                                                   &
    \end{tikzcd}
  \end{equation*}
  \noindent  commutes because
  \begin{equation*}
    \begin{aligned}
      &\mbar((\gamma^{(k)})^*_0 \cdot \iota_{Q}) \cdot
      \gamma^\#\\
      = &\mbar((\gamma^{(k)})^*_0) \cdot
      (\alpha \cdot \gamma)^*_0 \\
      = &(\gamma^{(k)})^*_1 \cdot (\alpha \cdot \gamma)^*_0 & \by{Lemma~\ref{lem:mBarKleisli}}\\
      = &((\gamma^{(k)})^*_1 \cdot (\alpha \cdot \gamma))^*_0 & \by{*}\\
      = &(\gamma^{(k+1)})^*_0 & \by{def. of $\gamma^{(k+1)}$}
    \end{aligned}
  \end{equation*}

  \noindent where (*) is by the standard associative law of the Kleisli
  star ($(f^*_n \cdot g)^*_m = f^{*}_{n+m} \cdot g^*_m$).
  
  So since
  $h\colon \mbar Q
  \to \mbar M_k1$ is initial, there must be, by the situation
  described in the diagrams, an arrow above
  $U(\mbar \iota_{Q}\cdot \gamma^\#)$, that is
  $\mbar \iota_{Q}\cdot \gamma^\#\colon P \to \mbar Q$ is a morphism
  in $\catE$. This
  implies that
  $P \sleq (\mbar \iota_Q \cdot \gamma^\#)^\bullet \mbar Q$.
  
  `$\Rightarrow$': Suppose that $P \in W_{k+1}$, i.e.\ Duplicator has a
  winning move $Z\subseteq B$ at~$P$. We again define $Q := \bigsqcup
  Z$. This move is admissible, i.e.
  $P \sleq (\mbar \iota_{Q} \cdot \gamma^\#)^\bullet \mbar C(
  Q)$, so
  $\mbar \iota_{Q} \cdot \gamma^\# \colon(M_0X,P)\to\mbar C(
  Q)$ is a morphism in~$\catE$. We have to show that
  $P \sleq (\gamma^{(k+1)})^*_0 M_{k+1}1$.  Since~$Z$ is winning,
  we have $P' \in W_k$ for all possible Spoiler moves $P' \in Z$,
  whence $P' \sleq ((\gamma^{(k)})_0^*)^\bullet M_k1$ by
  induction; therefore,
  $Q \sleq ((\gamma^{(k)})_0^*)^\bullet M_k1$. Thus we
  have a morphism
  $(\gamma^{(k)})^*_0 \colon C(Q) \to M_k1$ in
  $\catE$; that is, we have the situation in the following
  diagram (which commutes as already noted above):
  \begin{equation*}
    \begin{tikzcd}[column sep=large, row sep=large]
      P \arrow[r, "\gamma^\#"] \arrow[rrd,
      "(\gamma^{(k+1)})^*_0"'] & \mbar M_0X \arrow[r, "\mbar
      \iota_{Q}"] & \mbar C(Q)
      \arrow[d, "\mbar(\gamma^{(k)})^*_0"] \\
      & & M_{k+1}1 = \mbar M_k1&
    \end{tikzcd}
  \end{equation*}
  \noindent We know that the morphism on top and right are morphisms
  in $\catE$, and therefore so is their
  composite~$(\gamma^{(k+1)})^*_0$. This implies that
  $P \sleq ((\gamma^{(k+1)})^*_0)^\bullet M_{k+1}1$.
\end{proof}

\noindent Preservation of initial arrows is a crucial condition in the
proof of Theorem~\ref{thm:main-inf}, but this can be hard to show when the
codomain is an arbitrary algebra $(A,a)$, as it is not always easy to give a concrete description of the 
conformance on $\mbar (A,a)$. Restricting to arrows with codomain
$M_n1$ can simplify these types of proofs, since $\mbar M_n1 =
M_{n+1}1$.

The game can also be played without a bound on the number of rounds,
to characterize the infinite-depth graded behavioural conformance.
In this case, Duplicator wins infinite plays. This however requires
preservation of all initial arrows.

\begin{thm}
  \label{thm:main-inf}
  Let $(\alpha, \mathbb{M})$ be a depth-1 graded semantics for a
  functor $G\colon \catE \to \catE$ such that~$\mbar$ preserves
  initial arrows, and let $(X, \gamma)$ be a $G$-coalgebra.
  Duplicator wins a position $P \in B$ in the infinite graded
  conformance game iff
  $P \sleq P_{\gamma^\#}^\infty$.
\end{thm}
\begin{proof}
  ($\Rightarrow$) We prove the claim for an arbitrary $\mbar$ coalgebra,
  carried by an $M_0$-algebra $(A,a)$. Denote by $W_\omega \subseteq B$ the winning region
  of Duplicator in the infinite game. Then $P := C(\bigsqcup W_\omega)$ is
  a refinement of the
  $M_0$-algebra $(A,a)$, so we have an arrow $p\colon M_0
  P \to P$. We now need to show that $\mbar \iota_P \cdot \gamma^\#
  \colon P \to \mbar P$ is a morphism in $\mathcal{E}$, i.e.
  that the dashed arrow in the diagram below exists. This then
  implies the claim.
  \begin{equation*}
    \begin{tikzcd}[column sep=large, row sep=large]
      (A,a) \arrow[r, "\gamma^\#"] \arrow[d, "\iota_P"] & \mbar (A,a)
      \arrow[d, "\mbar \iota_P"]\\
      (P,p) \arrow[r, dashed, "\gamma^\#_{/P}"] & \mbar(P, p)
    \end{tikzcd}
  \end{equation*}

\noindent Let $Q\in W_\omega$. Then Duplicator has an admissible move $Z$ at
  $Q$, such that every answer by Spoiler is again winning for
  Duplicator, i.e. such that $Z \subseteq W_\omega$. By admissibility we
  have that $Q
  \sleq (\mbar \iota_Z \cdot \gamma^\#)^\bullet C(Z) \sleq (\mbar
  \iota_Z \cdot \gamma^\#)^\bullet P$. Therefore $P = C(\bigsqcup_{Q \in
  W_\omega} Q) \sleq (\mbar
  \iota_Z \cdot \gamma^\#)^\bullet P$. From this it follows that
  $\mbar \iota_P \cdot \gamma^\# \colon (P, a) \to \mbar(P, a)$ is a
  morphism in $\catE$.

  ($\Leftarrow$) Let $h\colon (A, \gamma^\#) \to (C, \delta)$ be a
  homomorphism of coalgebras. We have to show that any basis element
  below $Q := h^\bullet C$ is winning for Duplicator, i.e. that $(\downarrow
  h^\bullet C) \cap B \subseteq W_\omega$, where $\downarrow h^\bullet
  C$ denotes the downclosed subset of $\catE_{UA}$ generated by
  $h^\bullet C$. We show that Duplicator can maintain the invariant
  $(\downarrow h^\bullet C) \cap B$, i.e. for any move starting at a
  position $P \in (\downarrow h^\bullet C) \cap B$, Duplicator can force
  their next position to also be from $(\downarrow h^\bullet C) \cap B$.
  This is obviously the case if Duplicator plays $Z = (\downarrow
  h^\bullet C) \cap B$ as their move. It now only remains to show that
  $Z$ is admissible, that is $P \sleq (\mbar\iota \cdot
  \gamma^\#)^\bullet C(Z)$ for all $P \in Z$. Note that by
  Lemma~\ref{lem:hom-reindxing}, we have that $Q  = \bigsqcup Z$ is a
  refinement of $A$.

 \begin{equation*}
    \begin{tikzcd}[column sep=large, row sep=large]
      UQ =  UA \arrow[r, "U\gamma^\#"] \arrow[d, "Uh"] & U\mbar A
      \arrow[d, "U\mbar h"] \arrow[r, "U\mbar \iota"] &U \mbar Q
      \arrow[dl, "U\mbar \overline h"]\\
      UC \arrow[r, "\delta"] & U\mbar C
    \end{tikzcd}
 \end{equation*}

 \begin{equation*}
    \begin{tikzcd}[column sep=large, row sep=large]
      Q \arrow[r, dashed, "\mbar \iota \cdot \gamma^\#"] \arrow[d,
      "\overline h"] & \mbar Q
      \arrow[d, "\mbar \overline h"]\\
      C \arrow[r, "\delta"] & \mbar C
    \end{tikzcd}
 \end{equation*}

\noindent Consider the above diagrams, the top in $\Set$ and on the
bottom in
  $\catE$. From the initiality of $\mbar \overline h$, it follows that
  the dashed morphism $Q \to \mbar Q$ above $U(\mbar \iota \cdot
  \gamma^\#)$ exists, that is, that $Q \sleq
  (\mbar\iota \cdot \gamma^\#)^\bullet Q$, implying that
  $P \sleq \bigsqcup Z = Q \sleq (\mbar\iota \cdot
  \gamma^\#)^\bullet Q = (\mbar\iota \cdot \gamma^\#)^\bullet C(Z)$.
\end{proof}

\noindent The games above are played on conformances on the pre-determinization. To
determine whether a conformance holds in the base space $X$, one must first
transform it to a conformance on $M_0X$.

\begin{correp}
  Let the setting be as in Theorem~\ref{thm:main-fin} and
  Theorem~\ref{thm:main-inf}. Let $P \in \catE_{UX}$.
  \begin{enumerate}
    \item For $n<\omega$, Duplicator wins all positions
  $P' \in B$ such that $P' \sleq \eta_\bullet P$ in the graded
  $n$-round conformance game iff
  $P \sleq (\gamma^{(n)})^\bullet {M_n1}$.
    \item Duplicator wins  all
  positions $P' \in B$ such that $P' \sleq \eta_\bullet P$ in the
  infinite graded conformance game iff $P \sleq P^\infty_\alpha$.
  \end{enumerate}
\end{correp}
\begin{proof}
  Immediate by Theorem~\ref{thm:main-fin} and
  Theorem~\ref{thm:main-inf} and from the fact that $\eta_\bullet \dashv \eta^\bullet$.
\end{proof}

\begin{expl}
  \label{expl:games}
  \begin{enumerate}[wide]
    
  \item For trace inclusion semantics
    (Example~\ref{expl:graded-semantics}.\ref{item:trinc}), the
    functor~$\mbar$ can be written as $\mbar X = X^\act$ by distributing
    joins over actions, with the ordering on~$X^\act$ defined
    pointwise.  Since initiality on $\Pre$ corresponds to order
    reflection, it is follows that $\mbar f$ is initial whenever~$f$
    is.
    
    Let $\gamma\colon X \to \overline P_\omega(\act \times X)$ be a
    finitely branching labelled transition system. As the set
    $B \subseteq \Pre_{\Pow X}$ of basic conformances, we choose
    those preorders that relate precisely one pair of distinct
    elements of $\Pow X$.  To check trace inclusion $x\leq y$ for two
    states $x, y\in X$, we first need to transform this inequality to
    live on the carrier of the pre-determinization $\Pow X$ by forming
    the pushforward $\eta_\bullet P$, where $P$ is the preorder on $X$
    that is discrete, except for the inequality $x \leq y$.  The
    resulting preorder is the preorder on $\Pow X$ that is again
    discrete, except for the inequality $\{x\} \leq \{y\}$.

    The game then continues as follows: In position $A \leq B$, with sets
    $A, B \in \Pow X$, the one step behaviours of states in the
    determinization are sets
    $\gamma^\#(A) = \bigcup_{x\in A} \gamma(x) = \{(a_1,x_1), \ldots,
    (a_n,x_n)\}$ and
    $\gamma^\#(B) = \bigcup_{y\in B} \gamma(y) = \{(b_1, y_1), \ldots,
    (b_m,y_m)\}$. Duplicator then claims a set $Z$ of inequalities
    between sets consisting of the $x_i$ and $y_i$, respectively, that
    would imply $\gamma^\#(A) \leq \gamma^\#(B)$ in the ordering on
    $M_1X$, which intuitively means that
    for each $a \in \act$, the union of all trace sets of states $x_i$
    such that $a_i = a$ is included in the union of trace sets of states
    $y_i$ such that $b_i = a$. Spoiler can then challenge any of the claims
    $P\in Z$ made by Duplicator, by playing it as the new
    position. Note that Duplicator may restrict to playing one
    inequality $\{x_i\}\le S_i$ for each~$i$. Then, the order of play
    may in fact be reversed to let Spoiler move first as in the
    standard bisimulation game: Spoiler picks~$x_i$, and then
    Duplicator answers with~$S_i$.

    Note that Duplicator could in any position just claim that
    $S \leq \emptyset$ for some number of subsets $S \subseteq X$,
    which would allow them to easily construct an admissible
    move. Such a strategy would however lose after the final turn
    (unless really $S=\emptyset$), since calling the bluff would
    reveal that $\{*\} \not \le \emptyset$
  \item For probabilistic trace distance
    (Example~\ref{expl:graded-semantics}.\ref{item:qtrace}), %
    read elements of $\mbar A$ as distributions
    $\mu \in \overline \Prob(\act \times A)$ such that for each
    $a\in \act$, one has $\mu(a, x) \not = 0$ for at most one
    $x \in A$.  %
    Preservation of initial arrows by
    $\mbar$ then follows from preservation by $\overline \Prob$, which
    is well-known. Fix a
    labelled Markov chain
    $\gamma\colon X \to \overline \Prob (\act \times X)$. We first
    choose as the basis $B \subseteq \PMet_{\Prob X}$ the set of all
    metrics on $\Prob X$ that are discrete, except at one pair of
    distinct distributions. To play the game that tests for two states
    $x, y\in X$ having distance $\leq \epsilon$, we first construct
    the pushforward along $\eta$, giving us a metric that is discrete,
    except for $d(\mu_x, \mu_y) = \epsilon$, where
    $\mu_x(x) = \mu_y(y) = 1$. The one-step behaviour $\gamma^\#(\mu)$
    of a distribution $\mu \in \Prob X$ is a distribution on
    $\act \times X$, where
    $\gamma^\#(\mu)(a,x) = \sum_{y \in X} \mu(y)\gamma(y)(a, x)$.

    To play the game in the position that relates distributions $\mu$
    and $\nu$ at distance~$\delta$, Duplicator then claims a set of
    distances to hold between elements of $X$ that would imply
    $d(\gamma^\#(\mu), \gamma^\#(\nu)) \leq \delta$, when the distance
    is calculated as in $\overline \Prob(\act \times X)$. Spoiler then
    challenges one of the distances claimed by Duplicator and the game
    continues in this new position.
  \end{enumerate}
\end{expl}

\begin{toappendix}
\subsection*{Details for Example~\ref{expl:games}}
  We show explicitly that the functor $\mbar$ for
  $\mathbb{M}_{\ptrace}$ preserves initial morphisms with codomain
  $M_n1$. Let $(A,a)$ be a $\overline \Prob$ algebra, and let $h\colon
  A\to M_n1$ be an initial algebra homomorphism, that is, $d(x,y) =
  d(h(x), h(y))$ for all $x, y\in A$. Let $s,t\in \mbar A$, then the
  elements
  $s, t$ can be viewed as distributions of the form $\Prob
(\act \times A)$, where for all $a \in \act$ we have that $s(a, x)
\not = 0$ for at most one $x$. We need to show that $d(s, t)\leq
d(\mbar h(s), \mbar h(t))$. For this purpose, we are done once we show
that the distance $d(\mbar h(s), \mbar h (t))$ when calculated in the form $\overline \Prob(\act
\times \overline \Prob(\act^{n}))$ is equal to the
distance in the form $\Prob(\act ^{ n+1})$. Assume
that $\mu,\nu$ are given as a distribution of the prior form, and that
$\mu',
\nu'$ are given as the latter form. Then the distance $d(\mu,\nu)$ is
witnessed by a coupling $\rho\in \mathsf{Cpl}(\mu,\nu)$, which can be
viewed as a convex combination of the form $\sum_{I}\lambda_i((a_i, u_i), (b_i,
v_i))$, where $d(\mu, \nu) = \sum_{I}\lambda_i \max(d(a_i, b_i), d(u_i, v_i))$,
and for all $u_i,v_i\in \Prob(\act ^{\boxplus n})$ we also have that
the distance is witnessed by a coupling $\rho_i = \sum_{J_i}
\lambda_{i,j}(u_{i,j},
v_{i,j})$. From this data we can now construct a coupling of the
distributions $\mu',\nu'$, given by $\rho'=
\sum_{I,J_i}\lambda_i\lambda_{ij}((a_i, u_{i,j}), (b_i, v_{i,j}))$.
We then have:

\allowdisplaybreaks

\begin{align*}
  d(\mu, \nu) &= \sum_{I}\lambda_i \max(d(a_i, b_i), d(u_i, v_i))\\
          &= \sum_{I}\lambda_i \max\left( d(a_i, b_i), \sum_{J_i}
            \lambda_{i,j} d(u_{i,j},
            v_{i,j}) \right)\\
           &= \left( \sum_{i\in I, a_i \not = b_i} \lambda_i
           1 \right) + \sum_{i\in I, a_i =
         b_i}\lambda_i \sum_{j\in J_i}
            \lambda_{i,j} d(u_{i,j},
            v_{i,j})\\
           &= \left( \sum_{i\in I, a_i \not = b_i} \lambda_i
           \sum_{j\in J_i} \lambda_{i,j}1 \right) + \sum_{i\in I, a_i =
         b_i}\lambda_i \sum_{j\in J_i}
            \lambda_{i,j} d(u_{i,j},
            v_{i,j})\\
          &= \sum_{I} \left( \sum_{J_i}
            \lambda_i\lambda_{i,j}  \max(d(a_i, b_i), d(u_{i,j},
          v_{i,j})) \right)\\
          &\ge d(\mu',\nu')
\end{align*}

\noindent On the other hand, let $\rho' \in \mathsf{Cpl}(\mu',\nu')$ be the coupling
witnessing the distance $d(\mu',\nu')$, given by the convex combination
$\sum_{I} \lambda_i((a_i, u_i), (b_i, v_i))$. We construct a
coupling $\rho$ of $\mu, \nu$, given by the convex combination
$\sum_{a ,b\in \act} \lambda_{a,b}((a, u_{a,b}), (b, v_{a,b}))$
where
$\lambda_{a,b} = \sum_{a_i = a, b_i=b} \lambda_i$
and $u_{a,b}$ is
given by the formal convex combination $u_{a,b} = \sum_{a_i = a,
b_i=b} \frac{\lambda_i}{\lambda_{a,b}} u_i$ (and
symmetrically for $v_{a,b}$). We then have
\begin{align*}
  d(\mu', \nu') &= \sum_{i\in I}\lambda_i \max(d(a_i,u_i), d(b_i, v_i))\\
                &= \left( \sum_{a_i \not = b_i} \lambda_{i}1 \right) +
                \sum_{a_i = b_i} \lambda_i d(u_i, v_i)\\
                &= \left( \sum_{a \not = b} \lambda_{a,b}1 \right) +
                \sum_{a\in \act} \lambda_{a,a}\sum_{a_i = a}
                \frac{\lambda_i}{\lambda_{a,a}} d(u_i, v_i)\\
                &= \left( \sum_{a \not = b} \lambda_{a,b}1 \right) +
                \sum_{a\in \act} \lambda_{a,a}\sum_{a_i = a}
                \frac{\lambda_i}{\lambda_{a,a}} d(u_i, v_i)\\
                &\ge \left( \sum_{a \not = b} \lambda_{a,b}1 \right) +
                \sum_{a\in \act} \lambda_{a,a} d(u_{a,a}, v_{a,a})\\
                &= \sum_{a,b\in \act} \lambda_{a,b}
                \max(d(a,b),d(u_{a,b}, v_{a,b}))\\
                &\ge d(\mu,\nu)
\end{align*}
So, in conclusion, $d(\mu, \nu) = d(\mu', \nu')$.
\end{toappendix}

\section{Syntactic Perspectives via Relational Structures}
\label{sec:rel-algebra}

\noindent In previous work, restricted to the category of sets
\cite{DBLP:conf/lics/FordMSB022}, graded equational games were viewed
from a syntactic perspective, as witnessing a proof in graded equational logic: The one-step behaviours
$\gamma^\#(x), \gamma^\#(y) \in M_1X$ of states $x,y$ may each be
viewed as a term in a graded equational theory. Duplicator then plays
a set $Z$ of equalities of depth-0 terms, under which they are able to
prove that the one step behaviour of $x$ is equivalent to the one step
behaviour of $y$, i.e. $Z \vdash \gamma^\#(x) = \gamma^\#(y)$.
Spoiler then challenges one of the claimed equalities in~$Z$ and the
game continues in this new position. In our present setting, we extend
this syntactic perspective to topological categories given as
categories of relational structures, such as pseudometric spaces,
employing (for simplicity, a restricted version of) a previous
syntactic treatment of monads~\cite{DBLP:conf/calco/FordMS21} and
graded monads~\cite{DBLP:phd/dnb/Ford23} on such
categories. In our extended setting, game positions then are claims
about relations between states, $\pi(x_1, \ldots, x_n)$, with
Duplicator playing sets of such claims, implying that the current
position holds after taking one step in the coalgebra.

\subsubsection*{Categories of Relational Structures}\;By
\emph{relational structures}, we are more precisely referring to
structures given by any number of finitary relations, axiomatized by
infinitary Horn axioms. Formal definitions are as follows.
\begin{defn}
  \begin{enumerate}[wide]
  \item A \emph{relational signature} $\Pi$ consists of a set of
    \emph{relation symbols}, equipped with finite arities
    $\ar(\pi) \in \nat$ for $\pi \in \Pi$. A \emph{$\Pi$-edge} in a
    set~$X$ then has the form $\pi(x_1, \ldots, x_n)$ where
    $\pi\in\Pi$, $\ar(\pi) = n$, and $x_1,\dots,x_n \in X$. For a
    $\Pi$-edge $e = \pi(x_1, \ldots, x_n)$ in $X$ and a map
    $h\colon X \to Y$, we write $h \cdot e$ for the edge
    $\pi(h(x_1), \ldots, h(x_n))$ in $Y$. If $E$ is a set of edges
    over~$X$, then we define $h\cdot E = \{ h\cdot e \mid e \in
    E\}$. A \emph{$\Pi$-structure} is then a pair $(X, E)$ where~$X$
    is a set and~$E$ is a set of $\Pi$-edges in~$X$. We often call
    $\Pi$-structures just~$X$, and then write $E(X)$ for the set
    of edges of~$X$.  A morphism $g\colon (X, E)\to (Y, F)$ of
    $\Pi$-structures is a map $g\colon X \to Y$ that preserves edges,
    i.e.\ $g \cdot e \in F$ whenever $e\in E$. The collection of all
    $\Pi$-structures and their morphisms forms a category $\str(\Pi)$.

  \item Fix a set~$V$ of variables. A \emph{Horn clause} over~$\Pi$
    and~$V$ is a pair consisting of a set~$\Phi$ of $\Pi$-edges
    over~$V$ and a $\Pi$-edge~$\psi$ over~$V$, denoted
    $\Phi \Rightarrow \psi$. A \emph{Horn theory}
    $\mathcal{H} = (\Pi, \pazocal{A})$ then consists of a relational
    signature $\Pi$ and a set $\pazocal{A}$ of Horn clauses
    over~$\Pi$, which we refer to as the \emph{axioms}
    of~$\mathcal{H}$. A $\Pi$-structure $(X, E)$ \emph{satisfies} a
    Horn clause $\Phi \Rightarrow \psi$ if whenever
    $\kappa\cdot\Phi\subseteq E$ for a map $\kappa \colon V \to X$,
    then $\kappa \cdot \psi \in E$. Now let $\mathcal{H} $ be a Horn
    theory. A \emph{model} of $\mathcal{H}$ is a $\Pi$-structure
    $(X, E)$ that satisfies all axioms of~$\mathcal H$. We denote the
    full reflective subcategory of $\Str(\Pi)$ spanned by the models
    of $\mathcal{H}$ by $\Str(\mathcal{H})$.
  \end{enumerate}
\end{defn}

\begin{figure*}[h]
\begin{gather*} %
  (\mathsf{Mor})\;\frac{\{X \vdash_k \pi(t_{1j}, \ldots, t_{nj}) \mid j
      = 1, \ldots,
    \ar(\sigma)\}}{X \vdash_{\depth(\sigma)+k} \pi(\sigma(t_{11}, \ldots,
      t_{1\;\ar(\sigma)}), \ldots,
    \sigma(t_{\ar(\pi) 1}, \ldots, t_{\ar(\pi)\ar(\sigma)}))}\;(*)\\[1em]
  (\mathsf{RelAx})\;\frac{\{X \vdash_k \tau \cdot \phi \mid \phi \in \Phi
  \} }{X \vdash_k \pi(\tau(x_1), \ldots, \tau(x_n))}\; {(\Phi \Rightarrow \pi(x_1, \ldots,
  x_n)
\in \overline{\pazocal{A}}, \atop \tau\colon \var \to T_{\Sigma,k}(X))}\\[1em]
  (\mathsf{Ax})\;\frac{\{X \vdash_{k} \pi(\tau(y_1), \ldots, \tau(y_i))
      \mid \pi(y_1, \ldots, y_i) \in
  Y\}}{X \vdash_{m+k} \rho(\tau(x_1), \ldots, \tau(x_j))}\;(+)
  \\[1em]
  (\mathsf{Ctx})\;\frac{}{X \vdash_0 \pi(x_1, \ldots, x_n)}\;(\pi(x_1,
  \ldots, x_n) \in E(X)) 
\end{gather*}
\caption{The rules of deduction in graded relational logic}
\end{figure*}

\begin{expl}
   \label{expl:rel-categories}
   \begin{enumerate}[wide]
     \item For $\Pi = \emptyset$ and $\pazocal{A} = \emptyset$, the
       category $\Str(\mathcal{H})$ is isomorphic to $\Set$.
     \item \label{expl:preorders-relational} Let $\Pi = \{\leq\}$ and
       $\pazocal{A}$ consist of the clauses
        \begin{equation*}
          \Rightarrow x\leq x\qquad x \leq y, y \leq z \Rightarrow x\leq z
        \end{equation*}
        The category $\Str(\mathcal{H})$ is then isomorphic to the
        category $\Pre$ of preordered sets and order preserving maps.
      \item For
        $\Pi = \{{=_\epsilon} \mid \epsilon \in [0,1] \subseteq
        \mathbb{R}\}$, and $\pazocal{A}$ consisting of the axioms
        \begin{gather*}
          x =_\epsilon y \Rightarrow y =_\epsilon x \qquad \Rightarrow x =_0 x \qquad x =_\epsilon y \Rightarrow
          x=_{\epsilon \oplus \epsilon' } y \\
          x =_\epsilon y, y =_\delta z \Rightarrow x =_{\epsilon
          \oplus \delta} z \qquad
          \{x =_{\epsilon'} y \mid{\epsilon'} > \epsilon\} \Rightarrow x
          =_\epsilon y,
        \end{gather*}
        where $a\oplus b = \min(1, a+b)$, the category $\Str(\mathcal{H})$ is isomorphic to the category
        $\PMet$ of pseudometric spaces and nonexpansive
        maps.
   \end{enumerate}
 \end{expl}
\begin{rem}
  In comparison to the more general
  setting~\cite{DBLP:conf/calco/FordMS21}, we disallow equality in
  Horn clauses to ensure that $\Str(\mathcal{H})$ is topological
  over~$\Set$. Effectively, we thus exclude separation conditions;
  that is, we exclude metric spaces (which would need a Horn clause
  concluding $x=y$ from $x=_0y$) but allow pseudometric spaces, and we
  exclude partial orders (which need the antisymmetry axiom
  $\{x\le y, y\le x\}\Rightarrow x=y$) but allow preorders. Indeed,
  behavioural conformances do not usually satisfy such separation
  axioms; e.g.\ distinct states may be mutually similar or have
  behavioural distance~$0$.

  Rosick\'y~\cite{ROSICKY1981309} shows that models of Horn theories
  of this form are topological categories, and conversely that for a more
  general definition of Horn theory, admitting class-sized relational
  signatures and relational symbols of arbitrary cardinal-valued
  arity, all topological categories can be represented by Horn
  theories. For instance, the category of topological spaces and continuous maps,
  which does not have a small presentation, can be axiomatized via ultrafilter convergence.
\end{rem}

\subsubsection*{Graded Relational-algebraic Theories}\;
We next recall a syntactic presentation of graded monads on categories
of relational
structures~\cite{DBLP:conf/calco/FordMS21,DBLP:phd/dnb/Ford23}. For
simplicity, we keep to operations whose arities are just natural
numbers instead of more complicated objects, while contexts of axioms
will be objects of $\str(\Pi)$.
\begin{defn}
  A \emph{graded signature} consists of a set $\Sigma$ of operation
  symbols $\sigma$, each with an associated \emph{depth}
  $\depth(\sigma)$ and
  \emph{arity} $\ar(\sigma)$, which are both natural numbers.
  
    \end{defn}

\begin{defn}
  Let $\Sigma$ be a graded signature and~$X$ a set. We inductively
  define the sets $T_{\Sigma,n}(X)$ of \emph{$\Sigma$-terms of uniform
    depth $n$}:
  \begin{itemize}
    \item For $x\in X$, we have $x\in T_{\Sigma,0}(X)$.
    \item For $\sigma \in \Sigma$ with $\ar(\sigma) = n$ and $k\in
      \nat$,
      we have
      $\sigma(t_1, \ldots, t_n) \in T_{\Sigma,k+\depth(\sigma)}(X)$
      whenever $t_i\in T_{\Sigma,k}(X)$ for $i=1,\dots,n$.
  \end{itemize}
  A \emph{uniform depth-$k$ substitution} is a function $\tau \colon
  X\to T_{\Sigma,k}(Y)$. Then $\tau$ extends in the obvious way to a function on terms
  $\bar\tau_m\colon T_{\Sigma,m}(X) \to T_{\Sigma,m+k}(Y)$.

  \end{defn}

\noindent The axioms and judgements in our system of algebraic reasoning are then
given as relations in context:

\begin{defn}
  A $\emph{graded $\Sigma$-relation}$ (of \emph{depth~$k\in \nat$}) in
  \emph{context} $X$, where $X$ is a $\Str(\Pi)$-object, is
  syntactically of the form $X \vdash_k \pi(t_1, \ldots, t_n)$ where
  $\ar(\pi) = n$ and $\pi(t_1,\ldots ,t_n)$ is a
  $\Pi \cup \{ = \}$-edge in $T_{\Sigma,k}(X)$. A \emph{graded (relational-algebraic)
    theory} $\mathbb{T}=(\Sigma, Q)$ consists of a graded
  signature~$\Sigma$ and a set~$Q$ of graded $\Sigma$-relations, the
  \emph{axioms} of~$\mathbb{T}$.  \end{defn}

\noindent The full system of graded relational-algebraic
theories~\cite{DBLP:phd/dnb/Ford23} allows arities of operations to be
$\str(\Pi)$-objects, so that definedness of terms becomes an issue. In
a setting where all arities are just natural numbers, the proof system
deriving graded $\Sigma$-relations over $\mathbb{T}$ simplifies
significantly

We denote by $\overline{\pazocal{A}}$ the set of relational axioms $\pazocal{A}$,
extended with the axioms for $=$ (namely symmetry, reflexivity, and
transitivity).
The side condition~$(*)$ of $(\mathsf{Mor})$ is
$t_{ij}\in T_{\Sigma,k}(X)$ for all~$i,j$, i.e.\ just requires that
the terms that occur have uniform depth. The side condition~$(+)$ of
$(\mathsf{Ax})$ similarly requires that
$\tau\colon Y \to T_{\Sigma, k}(X)$ is a uniform-depth substitution
and moreover that $Y \vdash_{m} \rho (x_1, \ldots, x_j)$ is an axiom
of~$\mathbb{T}$. Rule $(\mathsf{Mor})$ embodies the semantic condition
that all operations are morphisms of $\Pi$-structures. Rule
$(\mathsf{RelAx})$ incorporates the Horn axioms of~$\mathcal{H}$, and
rule $(\mathsf{Ax})$ the axioms of the graded theory; both types of
axioms can be introduced in substituted form. Finally, rule
$(\mathsf{Ctx})$ allows using the edges of the context.

Each graded theory $\mathbb{T}$ then induces a graded monad
$\mathbb{M}_\mathbb{T}$ on $\Str(\mathcal{H})$, where $M_nX$ is the
set of depth-n $\Sigma$ terms in context $X$ modulo derivable
equality, equipped with the set of derivable $\Pi$-edges.
Multiplication collapses nested terms, and the unit converts variables
into terms~\cite{DBLP:conf/calco/FordMS21,DBLP:phd/dnb/Ford23}.

\begin{rem}
  \label{rem:syntactic-monads}
  \noindent Algebraic presentations of graded monads often allow for
  easy verification of most of the assumptions required in our
  framework by purely syntactic means. In particular, a graded monad
  is depth-1 if it is presented by a relational-algebraic theory
  containing only operations and axioms of depth at
most~$1$~\cite[Proposition 5.29]{DBLP:phd/dnb/Ford23}.  Moreover, we
have the following syntactic characterization of monads for which~$M_0$ is a lifting
  of a monad on $\Set$:
\end{rem}

\begin{proprep}
  Let $\mathbb{T}$ be a graded relational-algebraic theory, such that
  all depth-0 axioms that conclude equality have the form  $X \vdash_0
  t_1 = t_2 $, where $E(X)= \emptyset$. Then $M_0$ is a lifting of a
  monad on $\Set$.
\end{proprep}
\begin{proof}
  We show that every equality can be derived in equational logic,
  theories of which are well known to correspond to monads on $\Set$.
  We proceed by structural induction on the proof tree. Assume the
  root concludes with a judgement of the form $X \vdash_0 t_1 = t_2$.
  Since $E(X)$ does not contain edges of the form $x_1 = x_2$, the rule
  $(\mathsf{Ctx})$ can not be applied in the proof. Applications of
  the rule $\mathsf{(RelAx)}$ let us reason about symmetry,
  transitivity and reflexivity. They have premises of the form $t_1 = t_2$,
  for which the statement holds by induction. The rule is then
  applicable by definition in equational reasoning. For the rule
  $\mathsf{(Mor)}$, if $\sigma$ is $=$, then the statement also holds
  for all premises by induction, and is applicable in equational
  reasoning. Finally, if an application of
  the rule $\mathsf{(Ax)}$ concludes with an equality, then $E(Y) =
  \emptyset$, implying that the statement is vacuously satisfied.
\end{proof}

\begin{expl}
  \label{expl:graded-theories}
  \begin{enumerate}[wide]
  \item \label{item:branching-theory} Disregarding size issues for the
    moment, graded monads of the form $\mathbb{M}_G$ (Example
    \ref{expl:graded-monads}.\ref{item:gm-powers}), for some functor
    $G\colon \Str(\mathcal{H}) \to \Str(\mathcal{H})$, can be
    presented by taking as the set of operations $\Sigma$ (all of
    depth~$1$) the disjoint union of the sets $UG\kappa$, where
    $\kappa$ ranges over all cardinals (equipped with the discrete
    relational structure), with operation $\sigma \in UG\kappa$ having
    arity~$\kappa$.  Then~$\sigma$ is interpreted over $GX$ as the map
    that sends $f\colon \kappa \to X$ to $Gf(\sigma)$. The
    axiomatization then consists of all relations-in-context
    $X \vdash_1 \pi(t_1, \ldots, t_n)$ (where $\pi \in \Pi\cup \{=\}$
    and~$X$ is a $\Pi$-structure) such that $\pi(t_1, \ldots, t_n)$
    holds in $GX$.
  \item The graded monad $\mathbb{M}_\trinc$ on $\Pre$
    (Example~\ref{expl:graded-semantics}.\ref{item:trinc}) is presented by
    a relational algebraic theory containing a binary depth-0
    operator~$+$, a depth-0 constant~$\bot$, subject to depth-0 axioms
    asserting that these operators form a join semilattice structure
    and additionally,
    \begin{equation*}
      \vdash_0 \bot \leq x \quad \vdash_0 x \leq x + y \quad \{x \leq z, y \leq z\}
      \vdash_0 x +y \leq z
    \end{equation*}
    (cf.~\cite{DBLP:journals/mscs/AdamekFMS21})\lsnote{Check how this
      works over preorders}. Moreover, the theory includes unary
    depth-1 operations $a(-)$ for all $a\in\act$, along with the
    following depth-1 axioms:
    \begin{equation*}
      \vdash_1 a(\bot) = \bot \qquad \vdash_1 a(x+y) = a(x) + a(y)
    \end{equation*}
    By Remark~\ref{rem:syntactic-monads}, $\mathbb{M}_\trinc$ is thus
    depth-1 and $M_0$ is a lifting of a monad on sets.

  \item The monad $\mathbb{M}_\ptrace$ from
    Example~\ref{expl:graded-semantics}.\ref{item:pts} is presented by a relational
    algebraic theory with binary depth-0 operators $+_p$ for
    $p \in [0,1]$ and unary depth-1 operators $a(-)$ for $a\in
    \act$. The axioms of the theory then include the axioms of
    interpolative barycentric algebras (at depth~0), which present the
    monad $\overline \Prob$ in (non-graded) quantitative
    algebra~\cite{DBLP:conf/lics/MardarePP16}, and the depth-$1$ axiom
    \begin{equation*}
      \vdash_1 a(x +^p y) = a(x) +^p a(y)
    \end{equation*}
    Again, $\mathbb{M}_\trinc$ is depth-1 and~$M_0$ is a lifting of a
    set monad by Remark~\ref{rem:syntactic-monads}.
  \end{enumerate}
\end{expl}

\subsubsection*{Relational-algebraic view of graded games}\; In this
setting, i.e.\ when the graded semantics is presented as a
relational-algebraic theory, we obtain a more syntactic view of the
graded conformance game. Indeed, the game can be viewed as playing out
a proof in the system above: The basis of the lattice
$B\subseteq \Str(\mathcal{H})_{M_0X}$ can be chosen as all
conformances that are generated by a single edge in $M_0X$, so
positions of Duplicator may be viewed as single edges, while positions
of Spoiler can be viewed as sets of edges. Admissibility of a move
$Z$ in position $P$, where $P$ is represented by the edge $e$, then
boils down to provability of the judgement
$Z \vdash_1 \gamma^\# \cdot e$, that is, Duplicator must be able to
prove that the relation holds among the successor structures of the
current state, assuming the edges they claim in~$Z$ hold for the
successor states.

\begin{proprep}
  Given a relational-algebraic presentation of the monad $\mathbb{M}$,
  a move $Z$ of Duplicator is admissible in position $e$ iff $Z \vdash_1
  \gamma^\# \cdot e$.
\end{proprep}
\begin{proof}
  $(\Rightarrow)$ Assume that the move Z is admissible, i.e. that
  $\{e\} \sleq (\mbar \iota \cdot \gamma^\#)^\bullet \mbar C(Z)$. By
  adjoint transposition we have $(\mbar \iota \cdot \gamma^\#)_\bullet
  \{e\} \sleq \mbar C(Z)$, and by extension $\{\mbar \iota \cdot
  \gamma^\# \cdot e\} \sleq \mbar C(Z)$. By completeness, $\mbar
  C(Z)$ consists of all edges $e'$ in $M_1X$ such that $Z \vdash_1
  e'$, so we have that $Z \vdash_1 \mbar \iota \cdot \gamma^\# \cdot
  e$. Since $\mbar \iota$ only identifies terms that are provably
  equal, we can apply the rule $\text{(RelAx)}$ to arrive at the
  judgement $Z \vdash_1 \gamma^\# \cdot e$

  $(\Leftarrow)$ Assume that $Z \vdash_1 \gamma^\# \cdot e$. By
  soundness, we have that $\gamma^\# \cdot e \in E(\mbar C(Z))$,
  implying that $\{e\} \sleq (\mbar \iota \cdot \gamma^\#)^\bullet
  \mbar C(Z)$.
\end{proof}

\section{Examples}
\noindent Complementing the running examples, we now present two further case studies
in detail, bisimulation topologies and quantitative similarity in
metric LTS.

\subsubsection*{Bisimulation Topologies}\; The bisimulation topology
on a deterministic automaton $\gamma \colon X \to 2 \times X^\act$
equips the state space~$X$ with the topology generated by subbasic
opens of the form $\{x \in X \mid x \text{ accepts } \sigma \}$, for
all words $\sigma \in
\act^*$~\cite{DBLP:conf/lics/KomoridaKHKH19}. This conformance was
originally constructed by lifting the $\Set$-functor
$G = 2 \times -^\act$ to topological spaces via the codensity
construction. Since our framework lives natively in the topological
category, in this case $\Top$, we can use this lifting directly:
Define the functor $\overline G:\Top \to \Top$ as
$\overline G = 2 \times(-)^\Sigma$ where $2$ denotes the Sierpi\'nski space and~$X^\act$ carries the product topology. As our desired conformance is
branching-time, we consider the graded monad $\mathbb{M}_\btop$ given
by $M_n = \overline G^n$
(c.f. Example~\ref{expl:graded-monads}.\ref{item:gm-powers}), with
$\mu^{n,k} = \id$ and $\eta = \id$, and the graded semantics
$(\id, \mathbb{M}_\btop)$.

The graded monad $\mathbb{M}_\btop$ does not lend itself to a
relational-algebraic presentation, since $\Top$ is not a category of
(finitary) relational structures. The requirements of our framework
are nevertheless verified easily: Like all branching-time graded
monads (Example~\ref{expl:graded-monads}.\ref{item:gm-powers}),~$\mathbb{M}_\btop$ is
depth-1. Moreover, $M_0 = \Id$ lifts the identity functor on $\Set$
and preserves the terminal object. Preservation of initial arrows
follows from $\mbar = \overline G$ being a codensity lifting of a
certain form~\cite{DBLP:conf/cmcs/Komorida20}.

Recall that topological spaces on a set $X$ can be represented by
nearness relations of the form $R \subseteq X \times \Pow X$
satisfying certain axioms, where $(x, A) \in R$ iff $x$ is in the
closure of $A$. Then a function $f\colon X \to Y$ is continuous
between topological spaces $(X, R), (Y, S)$ if $(x,A) \in R$ implies
that $(f(x), f[A]) \in S$. We choose as the basis $B$ of the lattice
of topologies on  $X$ all topologies that are generated by singleton
relations $\{(x, A)\} \subseteq X \times \Pow X$. An element of the
form $(x, A)$ can then be interpreted in the topology on automata
described above: If $x$ accepts the word $\sigma$, then there is
a state $x'\in A$ that accepts $\sigma$.

In this setting, Duplicator plays
sets~$Z\subseteq X \times \Pow(X)$ of such claims, such that taking a step
in $\gamma$ implies the claim for $(x, A)$. This game differs markedly
from the codensity game~\cite{DBLP:conf/cmcs/Komorida20}, in which a
round begins with Spoiler playing a discontinuous modal predicate.

Graded semantics allows us to study this type of conformance also in
serial non-deterministic  automata, i.e.\ $G\Pow^*$-coalgebras, where $\Pow^*$
is the non-empty powerset functor: We lift the
monad~$\Pow^*$ to~$\overline \Pow^*\colon \Top \to \Top$ by equipping the set $\Pow^* X$ with the topology
generated by all sets $V_U$ for $U\subseteq X$ open, where
$V_U = \{ v \in \Pow^* X \mid v \cap U \not = \emptyset \}$. Easy
calculations show that unit and multiplication are continuous, i.e.
this lifting is indeed a monad~$\overline\Pow^*$ on $\Top$. We then have
a continuous distributive law
$\lambda\colon\overline\Pow^*\overline G \to\overline G\overline\Pow^*$, given by
\begin{alignat*}{2}
  &\pi_1 \cdot \lambda(T) &&= \textstyle\bigvee \Pow^* \pi_1(T)\\
  &\pi_2 \cdot \lambda(T)(a) &&=  \{ f(a) \mid f \in \Pow^*
                           \pi_2(T)\}
\end{alignat*}
which induces a graded monad $\mathbb{M}_\ndtop$ with
      $M_n = \overline G^n\overline\Pow^*$, as well as a graded semantics
$(\id, \mathbb{M}_\ndtop)$. Since $\mathbb{M}_\ndtop$ is induced by a
distributive law, it is depth-1 \cite{DBLP:conf/calco/MiliusPS15}, and
$\mbar$ is the lifting of the functor $2 \times (-)^\act$ from $\Top$
to $\EM(\overline\Pow^*)$, and hence preserves initial arrows.  The
pre-determinization~$\gamma^\#$ of a $\overline G\overline\Pow^*$-coalgebra
$(X,\gamma)$ corresponds to a topologization of the well-known
powerset construction. The graded conformance one obtains is a
bisimulation topology generated by subbasic opens
$\{x \in X \mid x \text{ accepts } \sigma \}$\lsnote{@Jonas: Add proof
  to appendix.}, now in
the language semantics of nondeterministic automata (not covered in
work on codensity games~\cite{DBLP:conf/lics/KomoridaKHKH19}). The
graded conformance game runs similarly as in the deterministic case
described above, now playing out in the determinized system ~$\gamma^\#$.

To illustrate the game more concretely, consider the nondeterministic
automaton $\gamma$ with state set $X= \{x_1,x_2,x_3\}$ over the
singleton alphabet $\act = \{a\}$ displayed in Figure~\ref{fig:automaton}.

\begin{figure}
  \begin{center}
\begin{tikzpicture}[shorten >=1pt,node distance=2cm,on grid,auto]
  \node[state, accepting]                     (x_1)                      {$x_1$};
  \node[state]                     (x_2) [below right=of x_1] {$x_2$};
  \node[state]                                (x_3) [below left=of x_1] {$x_3$};

  \path[->] (x_1) edge [bend left]      node        {a} (x_2)
            (x_1) edge [loop above]     node        {a} (x_1)
            (x_2) edge [bend left]      node        {a} (x_1)
            (x_1) edge                  node [swap] {a} (x_3)
            (x_3) edge [loop below]     node        {a} (x_3);
\end{tikzpicture}
\end{center}
\caption{A nondeterministic automaton over the alphabet $\Sigma = \{a\}$}
\label{fig:automaton}
\end{figure}

Let us assume that Duplicator wants to show the set $(x_2, \{x_1\})$
to be included in the bisimulation topology described above, i.e. that all
words of accepted by $x_2$ are also accepted by $x_1$. This set is
first transformed, via pushforward, into a claim on the powerset
construction, of the form $P = (\{x_2\},\{\{x_1\}\})$. Duplicator
must play an admissible set of basic claims $Z$, such that
$(\gamma^\#)^\bullet Z$ contains P. An example of such a move is the
singleton set containing the claim $(\{x_1\},\{\{x_1, x_2,
x_3\}\})$, for which admissibility is easily checked. Spoiler has no
option other than playing this claim, after which Duplicator may play
$(\{x_1, x_2, x_3\}, \{\{x_1, x_2, x_3\}\})$. In this Position the
moves repeat indefinitely, yielding the win for Duplicator. Duplicator
can always choose $Z$ to be a singleton set, since the alphabet
$\Sigma$ has only one element. For larger alphabets, Duplicator would
have to play one claim for each symbol, essentially allowing Spoiler to
challenge transitions. In this way a winning Spoiler
strategy in state $(x, A)$ spells out a distinguishing word,
accepted by $x$ but by no state in $A$.

\subsubsection*{Quantitative Similarity in Metric LTS}\;
We model quantitative simulation, which is
inherently asymmetric, by choosing as the topological category
$\QMet$, the category of hemimetric spaces and nonexpansive maps,
which is defined like $\PMet$, but without the symmetry axiom.
We denote by $\overline{\Pow}_\omega$ the
lifting of the finite powerset functor
$\pazocal{P}_\omega\colon \Set \to \Set$ to $\QMet$ given by equipping
with $\pazocal{P}_\omega X$ with the asymmetric Hausdorff
distance. That is, for $S, T \subseteq X$ we have
\begin{equation}\textstyle
d^\rightarrow(S, T) = \bigvee_{s \in S} \bigwedge_{t \in T} d(s, t).
\label{eqn:asym-hausdorff}
\end{equation}
We also consider the lifting $\overline \Pow^\leftrightarrow_\omega$ that equips
$\Pow X$ with the symmetric Hausdorff distance $d(S,T) =
d^\rightarrow(S,T)\vee d^\rightarrow(T,S)$.
Finitely
branching metric transition systems are coalgebras for the functor
$\overline{\Pow}^\leftrightarrow_\omega(\act \times {-})$, where $\act$ is a
(pseudo)metric space of labels. In analogy to the classical two-valued
linear-time / branching-time spectrum~\cite{Glabbeek01}, one has a
quantitative linear-time / branching-time spectrum on metric
transition systems~\cite{DBLP:journals/tcs/FahrenbergL14}, consisting
of behavioural metrics of varying granularity. In the following we will instantiate our
framework to both simulation distance and ready simulation distance.

We define the graded monad $\mathbb{M}_\qsim$ characterizing
(asymmetric) similarity distance by setting
$\mathbb{M}_\qsim = \mathbb{M}_G$ for $G = \overline \Pow_\omega (\act \times
{-})$ (recall that for $\mathbb{M}_G$ we
have $M_n = G^n$). Initiality preservation
of $\mbar = \overline{\Pow}_\omega (\act \times {-})$ then follows from the
  fact that the functor is an instance of the generalized Kantorovich
  construction. A theory for this graded monad can be constructed as
  in Example \ref{expl:graded-theories}.\ref{item:branching-theory}.
  \jfnote{Easier Algebraic theory?}

Since the graded monad is a variation of the branching-time monad,
determinization only has an effect on the distances, but not on the
structure of the system.
To play the game for quantitative simulation in a coalgebra
$\gamma\colon X \to \overline \Pow^\leftrightarrow_\omega(\act \times X)$, assume that
it is Duplicator's turn to play in position $x_1 =_\epsilon x_2$, where
$x_1, x_2\subseteq X$. Duplicator must then provide a
set $Z$ of distances between elements of $X$ such that
$d^\rightarrow(\gamma(x_1), \gamma(x_2)) \leq \epsilon$ follows under these
assumptions.
Spoiler then challenges one of those claims, and the game
continues in the next position.
Note that by~\eqref{eqn:asym-hausdorff} providing such a $Z$ amounts to finding
for each $(a,x) \in \gamma(x_1)$ a pair
$(a',x') \in \gamma(x_2)$ such that
$d(a,a')\leq \epsilon$, and adding the claim $d(x, x') \leq \epsilon$ to $Z$.
We may thus equivalently rephrase the moves in the game in a more
familiar manner, where in each round Spoiler first chooses an element of
$\gamma(x_1)$ and Duplicator only then responds with an element of
$\gamma(x_2)$. The game
relates strongly to existing quantitative
games~\cite{DBLP:journals/tcs/FahrenbergL14}, but pursues a more local
perspective: The graded conformance game is a standard win/lose game
and records explicit distance bounds in each position, while in
quantitative games~\cite{DBLP:journals/tcs/FahrenbergL14}, plays
receive quantitative reward values ex post.

 Similarly, we can %
define a graded monad $\mathbb{M}_\qrsim$ for ready simulation
distance by including at each step the available actions of the
current state, that is
$M_n = G^n$ for $G = \overline\Pow_\omega((\overline\Pow^\leftrightarrow_\omega\act)
\times \act \times {-})$.
The transformation $\alpha$ of the graded semantics $(\alpha,
\mathbb{M}_{\qrsim})$
is given by $\alpha(T) = \{(P_\omega
\pi_1(T), a, \{x\}) \mid (a,x) \in T\}$.

\jfnote{Theory for ready simulation.}

The game for ready simulation is then played like the simulation game, except that the distances of
currently enabled transitions factor into the calculation of the distance
of one-step behaviours $d(\gamma^\#(x_1),\gamma^\#(x_2))$.

\section{Conclusion and Future Work}

\noindent We have extended generic games for behavioural
equivalences~\cite{DBLP:conf/lics/FordMSB022} to cover system
semantics with additional structure given by a topological category,
such as order-theoretic, metric, or topological
structure. Parametrization over the granularity of the system
semantics relies on the framework of graded
semantics~\cite{DBLP:conf/calco/MiliusPS15,DBLP:conf/concur/DorschMS19},
in which the coalgebraic type functor is mapped into a suitably chosen
graded monad. We have provided a finite-depth and an infinite-depth
version of the game; the latter bear some resemblance to fixpoint
games on complete lattices, but are applicable only under additional
assumptions on the underlying graded monad. We have subsequently
instantiated the framework to the case where the topological category
determining the type of behavioural conformances is given as a
category of relational structures such as preorders or pseudometric
spaces. In this setting, we can view games as playing out proofs in
relational-algebraic
theories~\cite{DBLP:conf/calco/FordMS21,DBLP:phd/dnb/Ford23}, with
turns selectively expanding branches of the proof tree.

Issues for further research include the extraction of distinguishing
formulae from winning strategies, and the algorithmic calculation of
such winning strategies. We expect the presentation of topological
categories as categories of relational structures to play an important
role in the implementation of such algorithms.

\bibliographystyle{splncs04}
\bibliography{refs}

\end{document}